\documentclass{article}

\usepackage{lmodern}

\usepackage[margin=1.2in]{geometry}
\usepackage[utf8]{inputenc} % allow utf-8 input
\usepackage[T1]{fontenc}    % use 8-bit T1 fonts
\usepackage{hyperref}       % hyperlinks
\usepackage{url}            % simple URL typesetting
\usepackage{booktabs}       % professional-quality tables
\usepackage{enumitem}
\usepackage{amsmath}
\usepackage{amsfonts,amsthm,amssymb,comment}
\usepackage{leftindex}
\usepackage{nicefrac}       % compact symbols for 1/2, etc.
\usepackage{microtype}      % microtypography
\usepackage{mathabx}
\usepackage{mathtools}
\usepackage{xfrac}
\usepackage{csquotes}
\usepackage{algorithm}
\usepackage{algpseudocode}

\usepackage{bbm}

\usepackage{graphicx} 
\usepackage{caption}
\usepackage{subcaption}
\usepackage{enumitem}
\usepackage{cleveref}
\usepackage{bbm}
\usepackage{mathtools}
\usepackage{algorithm}
\usepackage{algpseudocode}
\usepackage{multirow}
\usepackage{enumitem}
\usepackage[]{xcolor}
\usepackage[numbers]{natbib}

\newtheorem{theorem}{Theorem}[section]

\newtheorem{proposition}[theorem]{Proposition}
\newtheorem{corollary}[theorem]{Corollary}
\newtheorem{lemma}[theorem]{Lemma}
\theoremstyle{remark}

\newtheorem{remark}{Remark}

\newtheorem{definition}{Definition}

\title{ADDIS-Graphs for online
error control with application to platform trials}

\author{Lasse Fischer\thanks{Competence Center for Clinical Trials Bremen, University of Bremen. Email: \texttt{fischer1@uni-bremen.de.}} \and Marta Bofill Roig\thanks{Center for Medical Data Science, Medical University of Vienna; Department of Statistics and Operations Research, Universitat Politècnica de Catalunya -- BarcelonaTech.} \and Werner Brannath\thanks{Competence Center for Clinical Trials Bremen, University of Bremen.}}

\begin{document}

\maketitle
\begin{abstract}
In contemporary research, online error control is often required, where an error criterion, such as familywise error rate (FWER) or false discovery rate (FDR), shall remain under control while testing an a priori unbounded sequence of hypotheses. The existing online literature mainly considered large-scale designs and constructed blackbox-like algorithms for these. However, smaller studies, such as platform trials, require high flexibility and easy interpretability to take study objectives into account and facilitate the communication. Another challenge in platform trials is that due to the shared control arm some of the p-values are dependent and significance levels need to be prespecified before the decisions for all the past treatments are available. We propose ADDIS-Graphs with FWER control that due to their graphical structure perfectly adapt to such settings and provably uniformly improve the state-of-the-art method. We introduce several extensions of these ADDIS-Graphs, including the incorporation of information about the joint distribution of the p-values and a version for FDR control.
\end{abstract}

\tableofcontents
\newpage

\renewcommand\thefootnote{\fnsymbol{footnote}}
\setcounter{footnote}{1}

\section{Introduction}\label{sec:intro}

In classical multiple testing $m\in \mathbb{N}$ hypotheses $H_1,\ldots, H_m$ are prespecified at the beginning of the evaluation. Modern data analysis, however, requires dynamic and flexible decision making. This led to the establishment of online multiple testing, where a potentially infinite stream of hypotheses $(H_i)_{i\in \mathbb{N}}$ is tested sequentially \citep{FS}. This means, at each step $i\in \mathbb{N}$, a decision is made on the current hypothesis $H_i$ while having access only to the previous hypotheses and decisions \citep{JM}. Since the number of future hypotheses is unknown as well, it is usually assumed to be infinite. Such online multiple testing problems can be found in many different research areas. Examples are platform trials \citep{Retal, zehetmayer2022online}, sequential modifications of machine learning algorithms \citep{FES, Fetal}, growing data repositories \citep{aharoni2014generalized, robertson2022online} and continuous A/B testing in the tech industry \citep{Ketal, RYWJ}. 

A widely known multiplicity adjustment in classical multiple testing is the Bonferroni correction, where an individual hypothesis $H_i$ is rejected, if its corresponding p-value $P_i$ is less than or equal to $\alpha/m$. Bonferroni's inequality immediately implies that the adjustment provides familywise error rate (FWER) control, where FWER is defined as the probability of rejecting any true null hypothesis. In a seminal paper, Holm (1979) \citep{holm1979simple} showed that the individual significance levels $\alpha/m$ can even be increased if some of the hypotheses are rejected. However, classical multiple testing procedures cannot be applied simply to online multiple testing.  The two previously mentioned procedures illustrate the difficulty of online multiple testing. First, the number of hypotheses $m$ in online testing is not prespecified in advance and could even be infinite. Second, data information about the other hypotheses can improve the multiple testing procedure, but in online multiple testing the individual significance level $\alpha_i$ for a hypothesis $H_i$ can only depend on the \textit{previous} p-values $P_1,\ldots, P_{i-1}$. 

While Bonferroni-like adjustments can be transferred to online multiple testing \citep{FS}, they usually lead to low power \citep{TR}. For this reason, Tian \& Ramdas (2021) \citep{TR} have established the following condition that can be used to prove FWER control for adaptive discarding (ADDIS) online multiple testing procedures:
\begin{align}
 \sum_{j=1}^{i} \frac{\alpha_j}{\tau_j-\lambda_j} (\mathbbm{1}\{P_j\leq \tau_j\}-\mathbbm{1}\{P_j\leq \lambda_j\}) \leq \alpha
 \quad \text{for all }i\in \mathbb{N},  \label{eq:cond_addis_principle}    
 \end{align}
where $\alpha_i\in (0,1)$, $\tau_i\in (\alpha_i,1]$ and $\lambda_i\in [0,1)$. The left-hand side of the upper inequality can be interpreted as the significance level spent up to step $i$. Hence, if $P_i>\tau_i$ or $P_i\leq \lambda_i$, the significance level $\alpha_i$ can be reused in the future testing process. The idea is that p-values corresponding to true hypotheses are often conservative and, therefore, tend to be large, and p-values corresponding to false hypotheses tend to be small such that many significance levels can be reused in the future. While \eqref{eq:cond_addis_principle} is helpful in showing that a given procedure controls the FWER, it is non-constructive and thereby of only little help for the construction of ADDIS procedures.

To control the FWER with condition \eqref{eq:cond_addis_principle}, $\alpha_i$, $\tau_i$ and $\lambda_i$ are only allowed to depend on information that is independent of $P_i$. 
% or at least only on information that $P_i$ is conditionally valid on.
The problem is that in order to exploit \eqref{eq:cond_addis_principle}, $\alpha_i$ needs to incorporate information about the previous values of $\mathbbm{1}\{P_j\leq \tau_j\}$ and $\mathbbm{1}\{P_j\leq \lambda_j\}$, $j<i$, and thus of the $p$-value $P_j$. Therefore, 
it is usually assumed that either all or at least some p-values are independent. 
Another issue that may limit the choice of individual significance levels is when the hypotheses are tested in an asynchronous manner \cite{ZRJ}. That means an individual significance level $\alpha_i$ needs to be determined before the corresponding $p$-value $P_i$ is observed. Hence, when $\alpha_i$ is to be determined, the information $\mathbbm{1}\{P_j\leq \tau_j\}$ and $\mathbbm{1}\{P_j\leq \lambda_j\}$ are only available for $p$-values $P_j$ where the testing process has already been completed. In general, for each hypothesis $H_i$ one can construct a conflict set \citep{ZRJ} that specifies on which of the previous $p$-values $\alpha_i$ cannot depend (e.g. due to asynchrony) or must not depend (e.g. due to dependence), as otherwise it would violate the FWER control. 

For example, in a platform trial many treatment arms $T_1,T_2,\ldots$ are compared to the same control group, however, not all treatment arms are in the platform from the beginning but added over time and the total number of treatment arms is unknown \citep{saville2016efficiencies} (see panel A of Figure \ref{fig:platform_trial} for an illustration). Hence, this can be interpreted as an online multiple testing problem \citep{Retal}. Throughout the paper, we assume that only concurrent controls are used, meaning for the evaluation of a treatment arm only control data of those patients is included that were randomised while the corresponding treatment arm was in the platform. This yields a local dependence structure of the p-values \citep{ZRJ}, since overlapping treatment arms share some control data, while p-values for non-overlapping treatment arms are independent. Furthermore, the individual significance level for a hypothesis needs to be determined when the treatment arm enters the platform, while the test decision is obtained when the treatment arm leaves the platform. Hence, the hypotheses are tested in an asynchronous manner.

For trials with multiple study objectives, Bretz et al. (2009) \citep{BWBP} proposed a graphical approach for FWER control to handle multiple test procedures in the classical multiple testing setting. In a graphical procedure, the hypotheses are represented by nodes which are connected by weighted vertices that illustrate the level allocation in case of a rejection. For example, the Bonferroni-Holm correction \citep{holm1979simple} for two hypotheses is illustrated in Figure \ref{fig:holm}. Initially, both hypotheses are tested at level $\alpha/2$. However, if one of the hypotheses is rejected, its level can be distributed to the remaining hypothesis such that it is tested at level $\alpha$.
This graphical representation has advantages like facilitating the illustration of study objectives and prioritisation of hypotheses. Robertson et al. (2020) \citep{RWB} extended this graphical approach to other error rates than FWER.

\begin{figure}[htbp]
	\begin{center}
			\centering
\includegraphics[width=14cm,height=4cm,keepaspectratio]{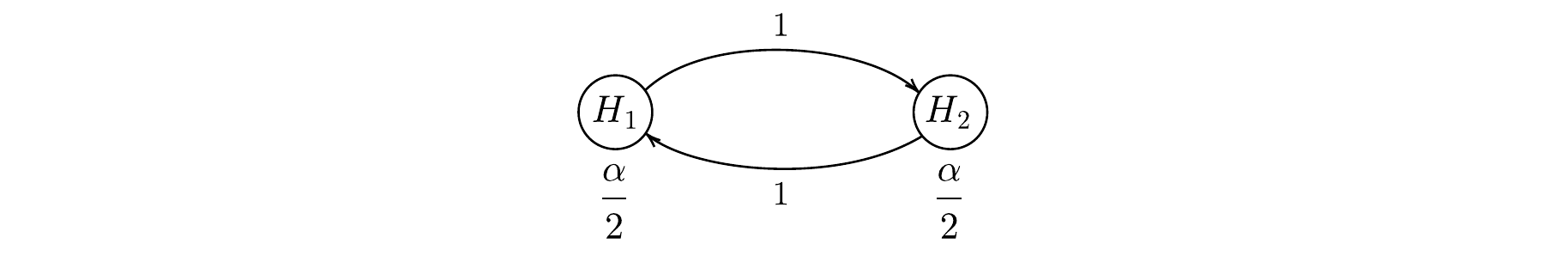}
	\end{center}
	\caption{The Bonferroni-Holm correction \citep{holm1979simple} represented as a graphical procedure \citep{BWBP}\label{fig:holm}}
\end{figure} 

Tian \& Ramdas (2021) \citep{TR} have introduced the ADDIS-Spending as a concrete algorithm satisfying the condition \eqref{eq:cond_addis_principle}. However, this algorithm does not exploit the full potential of the condition and is particularly inefficient under conflict sets. For this reason, in this paper, we propose the ADDIS-Graph, which is a constructive procedure that encompasses all possible online procedures satisfying condition \eqref{eq:cond_addis_principle} and uniformly  improves the ADDIS-Spending under conflict sets. Furthermore, it is a graphical procedure \citep{BWBP}, which is particularly useful for complex trial designs such as those needed in platform trials. For example, consider panel B of Figure \ref{fig:platform_trial}, where the platform trial is transferred into a graphical structure. The dotted lines represent the conflicts between the hypotheses/treatments arms. Due to its high flexibility, a graphical multiple testing procedure can easily adapt to this structure by distributing significance level only between hypotheses that are not connected. For example, the level of hypothesis $H_1$ may be distributed to hypotheses $H_3$, $H_4$ and $H_5$ but not to hypothesis $H_2$, and the level of hypothesis $H_2$ may be distributed to $H_5$ and $H_6$ but not to $H_3$ and $H_4$ (see panel C of Figure \ref{fig:platform_trial}).   

\begin{figure}[htbp]
	\begin{center}
			\centering
\includegraphics[width=42cm,height=14cm,keepaspectratio]{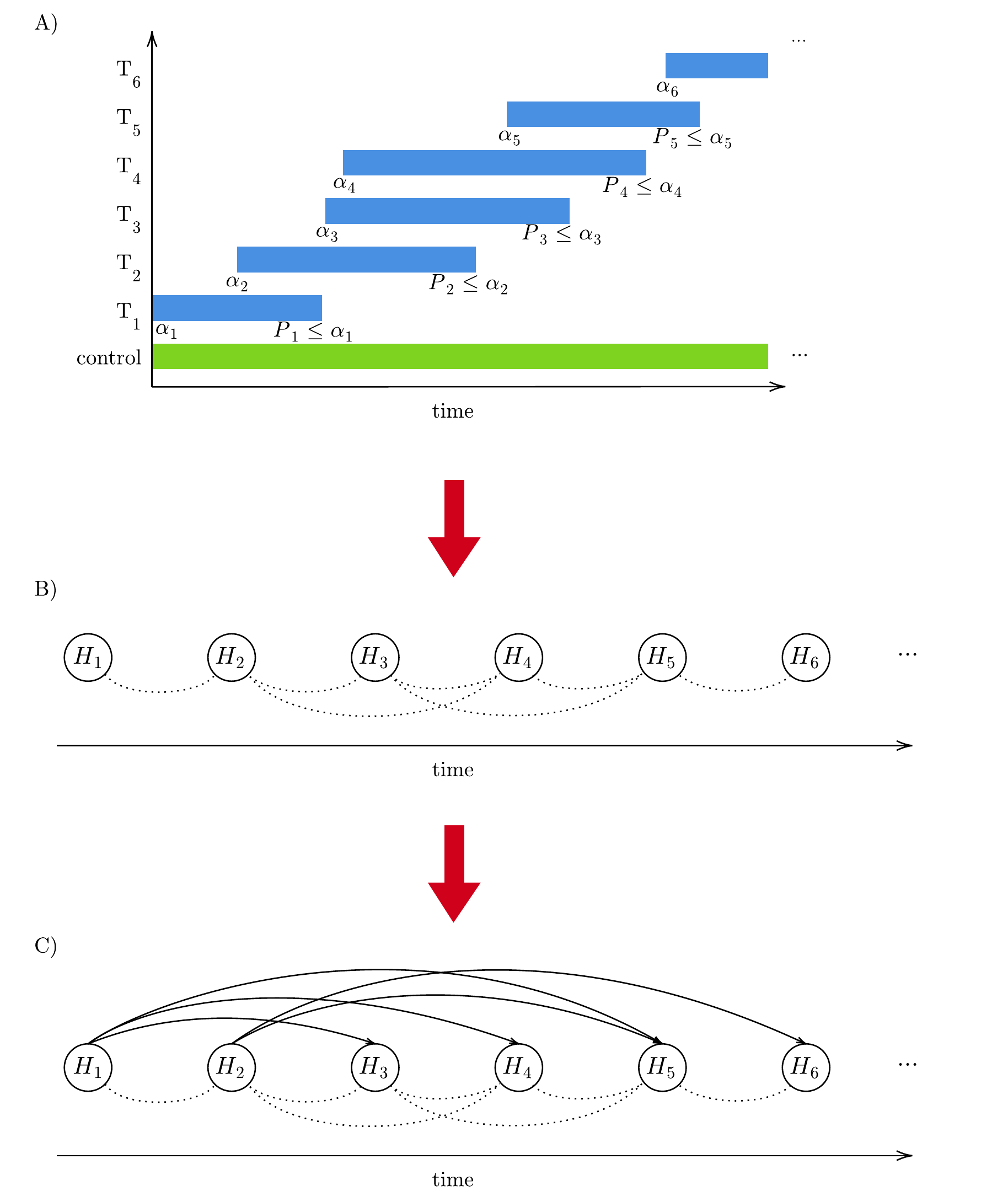}
	\end{center}
	\caption{Transferring a platform trial into a graphical procedure. The dotted lines indicate conflicts between the connected hypotheses. A graphical procedure can adapt to these conflicts by only distributing level between non-conflicting hypotheses. \label{fig:platform_trial}}
\end{figure} 

In platform trials, there has been a discussion about  which error rate is to be controlled. Some argue that in such clinical trials, strict control of FWER should be used and might also be a regulatory requirement \citep{Retal, wason2014correcting}; while others recommend controlling weaker error criteria, such as the false discovery rate (FDR), to avoid an increase of type II errors \citep{Retal, zehetmayer2022online}. 
The purpose of this paper is not to discuss which error rate is most appropriate but to construct multiple testing procedures that are powerful and easy to interpret in such complex online settings.  We focus on FWER control and, in addition, discuss extensions of our methods to other error rates in Section \ref{sec:extensions}.

\subsection{Overview of the paper}\label{sec:overview}

 We begin with a formal definition of the problem setting (Section \ref{sec:setting}).
 In Section \ref{sec:online_graph}, we derive the ADDIS-Graph when no conflicts are present and show that it contains all other online procedures satisfying condition \eqref{eq:cond_addis_principle}. In Section \ref{sec:local_dependence}, we adapt the ADDIS-Graph to conflict sets and prove that this leads to a uniform improvement over the $\text{ADDIS-Spending}$ under local dependence in Section \ref{sec:ADDIS-Spending}. 
 % Furthermore, we construct a new adaptive principle with a corresponding Adaptive-Graph (Section \ref{sec:correlation_structure}), which allows to exploit information about the joint distribution of $p$-values under local dependence, leading to further improvements. 
% Afterwards, we consider different examples to show how the ADDIS-Graph adapts to specific testing structures (Section \ref{sec:examples}).
Afterwards, we consider extensions of the ADDIS-Graph to other error rates and further improvements (Section \ref{sec:extensions}).
Finally, in Sections \ref{sec:sim} and \ref{sec:recovery_trial}, we demonstrate the application of the ADDIS-Graph through a simulation study and application to a real platform trial, respectively.  The \texttt{R}-Code for the simulations and case study is available at the GitHub repository \url{https://github.com/fischer23/Adaptive-Discard-Graph}. All formal proofs of the theoretical assertions are in the Appendix.

\section{Problem setting}\label{sec:setting}
 Let $I_0$ be the index set of true hypotheses, $R(i)$ be the index set of rejected hypotheses up to step $i\in \mathbb{N}$ and $V(i)=I_0\cap R(i)$ denote the index set of falsely rejected hypotheses up to step $i$. We aim to control the familywise error rate   $\text{FWER}(i)\coloneqq \mathbb{P}(|V(i)|>0)$ at each step $i\in \mathbb{N}$, where $\mathbb{P}$ denotes the probability under the true configuration of true and false hypotheses. Since $\text{FWER}(i)$ is nondecreasing, it is sufficient to control $\text{FWER}\coloneqq \mathbb{P}(v>0)$, where $v\coloneqq \lim\limits_{i\to \infty} |V(i)|$. The FWER is controlled strongly at level $\alpha$, if $\text{FWER}\leq \alpha$ for any configuration of true and false null hypotheses. 
% In contrast, weak control only provides that $\text{FWER}\leq \alpha$ under the global null hypothesis, which assumes that all hypotheses are true ($I_0=\mathbb{N}$). In this paper, we focus on strong control.
We assume that 
% Each $p$-value $P_i$, $i\in \mathbb{N}$, is a random variable $P_i:(\Omega, \mathcal{A}, \mathbb{P})\longrightarrow ([0,1], \mathcal{B})$, where $\mathcal{B}$ is the Borel $\sigma$-algebra. 
each null $p$-value $P_i$, $i\in I_0$, is  valid, meaning $\mathbb{P}(P_i\leq x)\leq x$ for all $x\in [0,1]$. A hypothesis $H_i$ is rejected, if $P_i\leq \alpha_i$, where $\alpha_i\in [0,1)$ is the individual significance level of $H_i$. In order to apply a multiple testing procedure in the online setting, the individual significance levels are only allowed to depend on the previous $p$-values. Mathematically, $\alpha_i$, $i\in \mathbb{N}$, is measurable with respect to the sigma algebra $\mathcal{G}_{i-1}\coloneqq\sigma(\{P_1, \ldots,P_{i-1}\})$.

As in Zrnic et al. (2021) \citep{ZRJ} we define $\mathcal{X}_i\subseteq \{1,\ldots,i-1\}$ as the index set of previous hypotheses conflicting with $H_i$. The conflict set $\mathcal{X}_i$ includes all indices of previous p-values that are \emph{not} stochastically independent of $P_i$, but can also contain further indices due to asynchrony or other restrictions. It is also assumed that the conflict sets $(\mathcal{X}_i)_{i\in \mathbb{N}}$ are monotone \citep{ZRJ}, which means that $j\in \mathcal{X}_i$ implies $j\in \mathcal{X}_k$ for all $k\in \{j+1,\ldots, i-1\}$. This ensures that the information we are allowed to use at each step $i\in \mathbb{N}$ is not decreasing over time. For example, this is fulfilled in every platform trial (e.g. Figure \ref{fig:platform_trial}), since an overlap between $T_j$ and $T_i$, $j<i$, implies that $T_j$ and $T_k$ overlap for all $j<k<i$ as well (if we order the treatments by entry time). Furthermore, each  $\mathcal{X}_i$ can be considered as fixed, although it might depend on information that is independent of $P_i$. In case of $\mathcal{X}_i=\emptyset$ for all $i\in \mathbb{N}$, we speak of trivial conflict sets. In order to conclude FWER control from condition \eqref{eq:cond_addis_principle} \citep{TR}, $\alpha_i$, $\lambda_i$ and $\tau_i$ must be measurable with respect to $\mathcal{G}_{-\mathcal{X}_i}\coloneqq \sigma(\{P_j:j<i, j\notin \mathcal{X}_i\})$.
% and the null p-values must be conditionally super-uniform (CS), meaning $\mathbb{P}(P_i\leq x|\mathcal{G}_{-\mathcal{X}_i})\leq x$ for all $i\in I_0,x\in [0,1]$.
Furthermore, in case of $\tau_i<1$, the null $p$-values $P_i$, $i\in I_0$, are required to be uniformly valid, meaning  $\mathbb{P}(P_i\leq xy|P_i\leq y) \leq x $ for all $x,y\in [0,1]$. However, this condition is fulfilled in many of the usual testing problems \citep{ZSS}.

\section{ADDIS-Graph under trivial conflict sets}\label{sec:online_graph}
We start with the construction of ADDIS-Graphs under trivial conflict sets, which means $\alpha_i$, $\tau_i$ and $\lambda_i$ need to be measurable with respect to $\mathcal{G}_{i-1}$. For this, we define $S_i\coloneqq \mathbbm{1}\{P_i\leq \tau_i\}$ and $C_i\coloneqq \mathbbm{1}\{P_i\leq \lambda_i\}$.

% Tian and Ramdas (2021)\citep{TR} showed by means of simulations that under trivial conflict sets the ADDIS-Spending  \eqref{eq:addis_spending} leads to a substantially higher power than the one achieved when using the Online-Graph \eqref{eq:online_graph}. However, the Online-Graph better adapts to complex testing situations. To this end, we bring together the approaches of the Online-Graph   and the ADDIS principle (Theorem \ref{theo:addis_principle}), resulting in what we call the ADDIS-Graph. We start with trivial conflict sets, meaning the hypotheses are tested in a usual online setting without asynchrony and the $p$-values are independent.

\begin{definition}[ADDIS-Graph under trivial conflict sets\label{def:addis_graph}]
Let $(\gamma_i)_{i\in \mathbb{N}}$ and $(g_{j,i})_{i =j+1}^{\infty}$, $j\in \mathbb{N}$, be non-negative sequences that sum to at most one. In addition, let $\tau_i\in (0,1]$ and $\lambda_i \in [0,\tau_i)$ be measurable regarding $\mathcal{G}_{i-1}$ for all $i\in \mathbb{N}$. The \textit{ADDIS-Graph} tests each hypothesis $H_i$ at significance level 
\begin{align}
\alpha_i = (\tau_i-\lambda_i)\left(\alpha \gamma_i + \sum_{j=1}^{i-1} g_{j,i}(C_j-S_j+1)  \frac{\alpha_j}{\tau_j-\lambda_j}\right). 
\label{eq:addis_graph}\end{align}
\end{definition}
\begin{theorem}\label{theo:addis_graph}
The ADDIS-Graph satisfies the condition \eqref{eq:cond_addis_principle} and thus controls the $\text{FWER}$ in the strong sense under the setting in Section \ref{sec:setting} when the conflict sets are trivial.
\end{theorem}

In order to represent this ADDIS-Graph as a graph, consider $\tilde{\alpha}_i=\alpha_i\frac{1}{\tau_i-\lambda_i}$ for all $i\in \mathbb{N}$, where $\alpha_i$ is the significance level obtained by the ADDIS-Graph. Equation \eqref{eq:addis_graph} gives us $\tilde{\alpha}_i=\alpha_i\frac{1}{\tau_i-\lambda_i}=\alpha \gamma_i + \sum_{j=1}^{i-1} g_{j,i}(C_j-S_j+1) \tilde{\alpha}_j$. 
Therefore, $\tilde{\alpha}_i=\alpha \gamma_i$ at the beginning of the testing process and in case of $P_j\leq \lambda_j$ or  $P_j>\tau_j$, the future levels $\tilde{\alpha}_i$, $i>j$, are updated by $\tilde{\alpha}_i=\tilde{\alpha}_i+g_{j,i} \tilde{\alpha}_j$. This is illustrated in Figure \ref{abb:addis_graph}, where the hypotheses are represented by nodes that are connected by weighted vertices indicating the error propagation, like in the classical graphical procedure \citep{BWBP}. The rectangles below the nodes clarify that the level $\tilde{\alpha}_i$ needs to be multiplied with $(\tau_i-\lambda_i)$ before comparing it with the $p$-value $P_i$. Note that this testing factor is only multiplied with the individual significance level when the corresponding hypothesis is tested, but it is not involved in the updating process with the graph.

\begin{figure}[htbp]
	\begin{center}
			\centering
			\includegraphics[width=15cm,height=5cm,keepaspectratio]{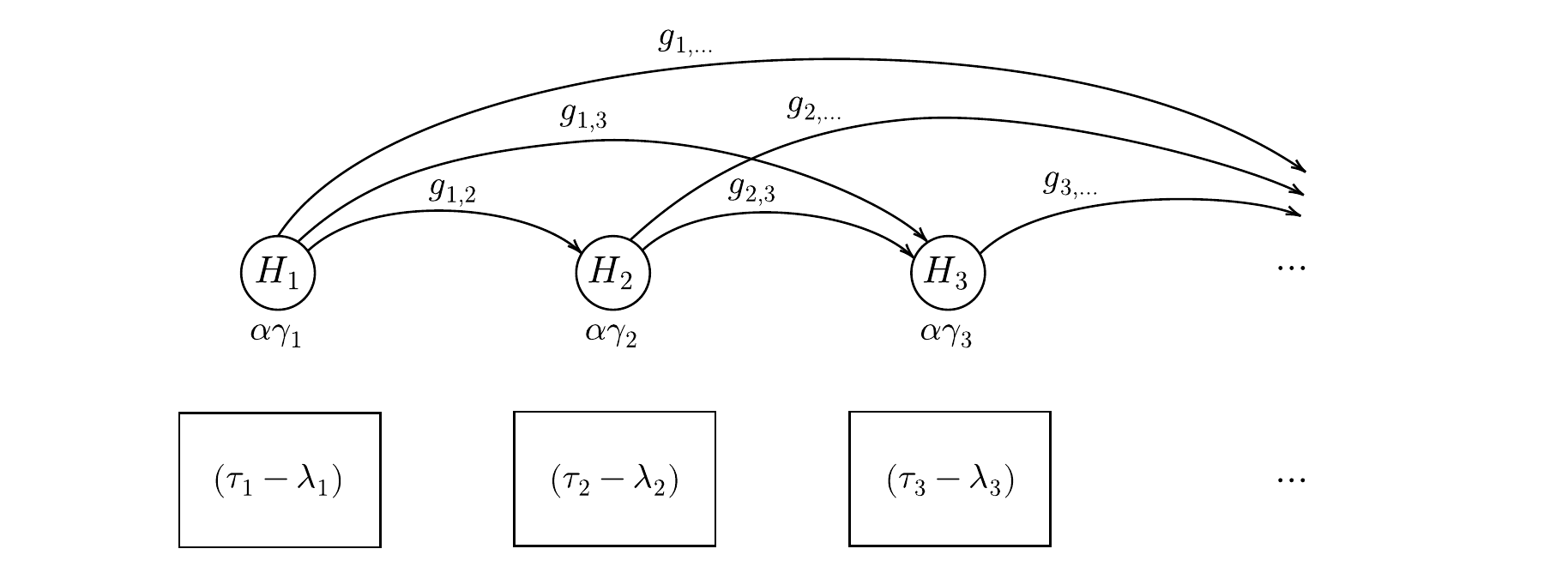}
	\end{center}
	\caption{Illustration of the ADDIS-Graph. Ignoring the rectangles the figure can also be interpreted as the Online-Graph. \label{abb:addis_graph}}
\end{figure}

In Definition \ref{def:addis_graph}, we considered the parameters $(\gamma_i)_{i\in \mathbb{N}}$ and $(g_{j,i})_{i=j+1}^{\infty}$, $j\in \mathbb{N}$ as fixed. However, $\gamma_i$ and $g_{j,i}$ could also be random variables that are measurable regarding $\mathcal{G}_{i-1}$, since we only need to ensure that $\alpha_i$ is measurable with respect to $\mathcal{G}_{i-1}$.
With this, the procedures become more flexible.  It can even be shown that, in this case, the ADDIS-Graph is the general ADDIS procedure, thus containing all online procedures satisfying condition \eqref{eq:cond_addis_principle}.

\begin{proposition}\label{theo:general_addis_graph}
Let $\gamma_i$ ($i\in \mathbb{N}$) and $g_{j,i}$ ($j\in \mathbb{N}$, $i>j$)  be measurable with respect to $\mathcal{G}_{i-1}$. Then, any online procedure satisfying condition \eqref{eq:cond_addis_principle} can be written as an ADDIS-Graph (Definition \ref{def:addis_graph}).
\end{proposition}

Note that Proposition \ref{theo:general_addis_graph} is not restricted to trivial conflict sets. Thus, if an online procedure $(\alpha_i)_{i\in \mathbb{N}}$ was adapted to conflict sets $(\mathcal{X}_i)_{i\in \mathbb{N}}$ such that $\alpha_i$ is measurable with respect to $\mathcal{G}_{-\mathcal{X}_i}$, it can also be constructed as an ADDIS-Graph that is given by \eqref{eq:addis_graph}. Therefore, being an ADDIS-Graph is necessary for a FWER controlling online procedure satisfying condition \eqref{eq:cond_addis_principle} under any conflict sets. Theorem \ref{theo:addis_graph} implies that being an ADDIS-Graph is also sufficient to control the FWER with \eqref{eq:cond_addis_principle} under trivial conflict sets. In the following section, we introduce a smaller class of ADDIS-Graphs that are sufficient for FWER control  under monotone conflict sets.

\section{ADDIS-Graph under nontrivial conflict sets}\label{sec:local_dependence}

% In this section, we show how the ADDIS-Graph \eqref{eq:addis_graph} can be adjusted to preserve FWER control under conflict sets. Afterwards, we consider the special case of monotone conflict sets (Section \ref{sec:mon_conflict}). This is the usual type of conflict sets considered in existing literature\citep{TR, ZRJ} and , for example, satisfied in platform trials. We show that the ADDIS-Graph uniformly improves the ADDIS-Spending under monotone conflict sets. In Section \ref{sec:non_mon_conflicts}, we outline the technical difficulties that should be considered when working with non-monotone conflict sets.

The ADDIS-Graph (Figure \ref{abb:addis_graph}) can easily be adjusted to control the FWER under conflict sets by removing arrows connecting conflicting hypotheses.
% A simple way to account for local dependence in the ADDIS-Graph (Figure \ref{abb:addis_graph}) is to remove the arrows connecting dependent $p$-values and adjust the individual significance levels of the ADDIS-Graph (Definition \ref{def:addis_graph}) to
% \begin{align*}\alpha_i = (\tau_i-\lambda_i)\left(\alpha \gamma_i + \sum_{j=1}^{i-L_i-1} g_{j,i} (C_j-S_j+1)  \frac{\alpha_j}{\tau_j-\lambda_j}\right).\label{eq: level_addis_fall_loc}\end{align*}
This is illustrated in Figure \ref{fig:addis_graph_loc} for a specific conflict set. In this example, $\mathcal{X}_2=\{1\}$, meaning $H_1$ and $H_2$ are conflicting.  Hence, the link $g_{1,2}$ is removed as no significance level of the first hypothesis can be allocated to the second. Note that by removing the weight $g_{1,2}$ potential significance level is lost. However, the ADDIS-Graph allows to enlarge the remaining weights due to this loss. For example, by adding $g_{1,2}$ to the weight $g_{1,3}$. By this, the same amount of significance level is distributed as in the case of trivial conflict sets, ensuring similar power. In general, the idea is to receive significance level for hypothesis $H_i$ only from hypotheses that are not contained in $\mathcal{X}_i$. 
% In panel C of Figure \ref{fig:platform_trial} it is illustrated how an ADDIS-Graph could be adapted to the local dependence structure of the given platform trial. 

\begin{figure}[htbp]
	\begin{center}
			\centering
			\includegraphics[width=15cm,height=5cm,keepaspectratio]{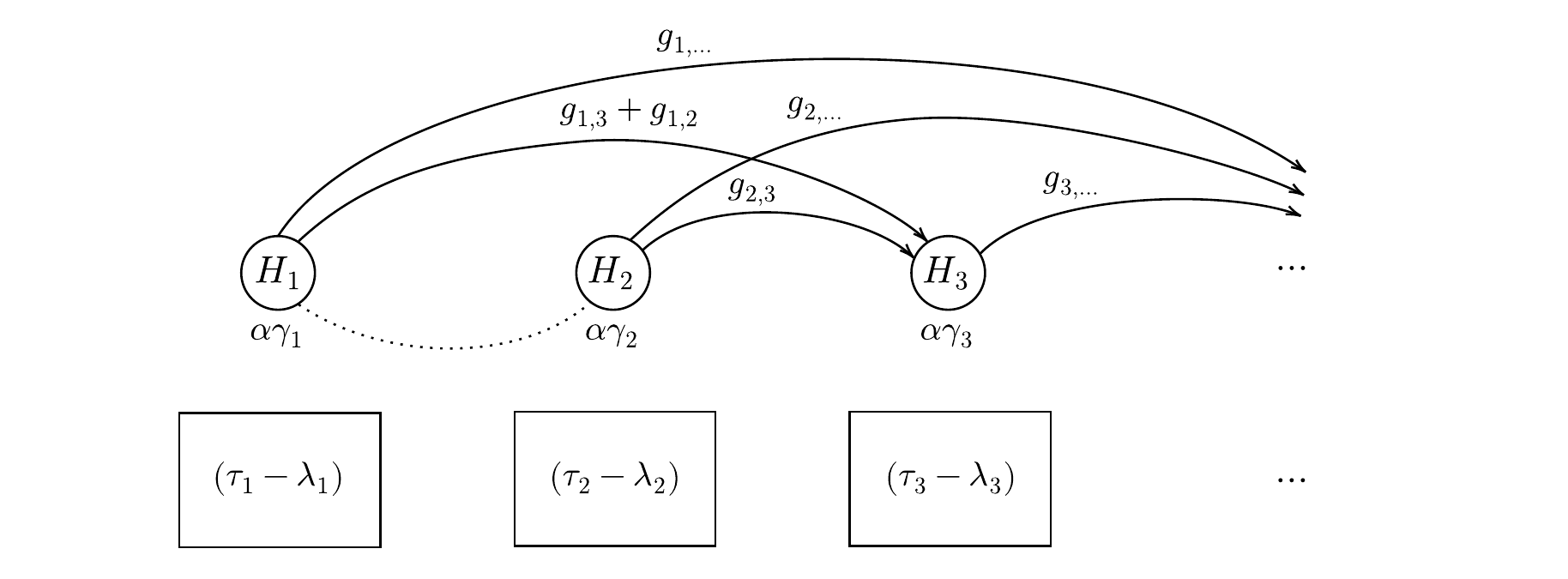}
	\end{center}
	\caption{Possible adjustment of the ADDIS-Graph (Figure \ref{abb:addis_graph}) to the conflict set $\mathcal{X}_2=\{1\}$.\label{fig:addis_graph_loc}}
\end{figure}

% In order to generalise such strategies, note that $L_{i+1} \leq L_i +1$ for all $i\in \mathbb{N}$ implies that $i-L_i$ is nondecreasing in $i$. Hence, if we define $d_j \coloneqq \min \{i\in \mathbb{N}: i-L_i >j\}$ (we set $\min(\emptyset)=\infty$) as the index of the first future $p$-value that does not depend on $P_j$, all $P_k$ with $k\geq d_j$ are independent from $P_j$ as well. Thus, the idea is to distribute the entire significance level of $H_j$ in case of $P_j \leq \lambda_j$ or $P_j > \tau_j$ only to hypotheses $H_k$ with $k\geq d_j$. 

% For this, we propose to remove the weights between dependent hypotheses and standardise the  remaining weights, which leads to the following ADDIS-Graph under local dependence.

 \begin{definition}[$\text{ADDIS-Graph}_{\text{conf}}$\label{def:addis_graph_loc}]
 Let the conflict sets be given by $(\mathcal{X}_i)_{i\in \mathbb{N}}$. Let $(\gamma_i)_{i\in \mathbb{N}}$ be a non-negative sequence that sums up to $1$ and $(g_{j,i}^*)_{i =j+1}^{\infty}$ be a non-negative sequence for all $j\in \mathbb{N}$ such that $g_{j,i}^*=0$ if $j\in \mathcal{X}_i$ and $\sum_{i>j, j\notin \mathcal{X}_i} g_{j,i}^* \leq 1$. In addition, let $\gamma_i$, $g_{j,i}^*$, $\tau_i\in (0,1]$ and $\lambda_i\in [0,\tau_i)$ be measurable regarding $\mathcal{G}_{-\mathcal{X}_i}$. The $\textit{ADDIS-Graph}_{\textit{conf}}$ tests each hypothesis $H_i$ at significance level 
 
 \begin{align}\alpha_i = (\tau_i-\lambda_i)\left(\alpha \gamma_i + \sum_{j=1}^{i-1} g_{j,i}^* (C_j-S_j+1)  \frac{\alpha_j}{\tau_j-\lambda_j}\right).\label{eq:addis_graph_loc}\end{align}
 
 % where $g_{j,i}^* = g_{j,i}\bigg/ \left(1-\sum_{k=j+1}^{d_j-1} g_{j,k}\right)$ if $i \geq d_j$ and $g_{j,i}^*=0$ otherwise.
 \end{definition}
 
Note that $\alpha_i$ from \eqref{eq:addis_graph_loc} is measurable with respect to $\mathcal{G}_{-\mathcal{X}_i}$, since $g_{j,i}^*=0$ for all $j\in \mathcal{X}_i$ and the conflict sets $(\mathcal{X}_i)_{i\in \mathbb{N}}$ are monotone. With this, the FWER control of $\text{ADDIS-Graph}_{\text{conf}}$ comes directly by Theorem \ref{theo:addis_graph}. 
\begin{corollary}
    The $\text{ADDIS-Graph}_{\text{conf}}$ controls the FWER in the strong sense under the setting in Section \ref{sec:setting} when conflicts are present.
\end{corollary}

Also note that for $\mathcal{X}_i=\emptyset$ for all $i\in \mathbb{N}$ the $\text{ADDIS-Graph}_{\text{conf}}$ becomes the ADDIS-Graph under trivial conflict sets (Definition \ref{def:addis_graph}). The $*$ in $(g_{j,i}^*)_{j\in \mathbb{N}, i>j}$ indicates that the weights are adjusted to the conflict sets. There are many approaches to obtain adjusted weights $(g_{j,i}^*)_{j\in \mathbb{N}, i>j}$. For example, one could choose them manually based on contextual information or add the removed weights to one of the remaining weights as in Figure \ref{fig:addis_graph_loc}. In the following section, we introduce a choice of weights that leads to a uniform improvement of the ADDIS-Spending \citep{TR}.

\section{A uniform improvement of the ADDIS-Spending}\label{sec:ADDIS-Spending}
The current state-of-art method satisfying condition \eqref{eq:cond_addis_principle} is the ADDIS-Spending by 
Tian \& Ramdas (2021) \citep{TR}, which controls the FWER under local dependence. The $p$-values $(P_i)_{i\in \mathbb{N}}$ are said to be locally dependent \citep{ZRJ}, if 
\begin{align}
    P_i \perp P_{i-L_i-1}, P_{i-L_i-2}, \ldots, P_1
\end{align}
for some lags $(L_i)_{i\in \mathbb{N}}$ with $L_{i+1}\leq L_i+1$. Zrnic et al. (2021) \citep{ZRJ} noted that local dependence is included in the more general concept of conflict sets by defining $\mathcal{X}_i=\{i-1,\ldots, i-L_i\}$ while the condition $L_{i+1}\leq L_i+1$ ensures that these conflict sets are monotone. For example, an intuitive special case of local dependence is batch dependence. That means, there are disjoint groups of $p$-values $B_1=\{P_1,\ldots,P_j\}$ for $j\in \mathbb{N}$, $B_2 =\{P_j,\ldots,P_k\}$ for $k>j$, $B_3= \{P_k,\ldots,P_l\}$ for $l>k$, and so on, such that $p$-values from the same batch may depend on each other, but hypotheses from different batches are independent. For instance, batch dependence occurs when the data is replaced by fresh and independent data after a period of time.

The $\textit{ADDIS-Spending}_{\textit{local}}$ is defined as
\begin{align}\alpha_i^{\text{spend}} = \alpha (\tau_i-\lambda_i) \gamma_{t(i)^{\text{loc}}}, \quad \text{where } t(i)^{\text{loc}}=1+ L_i +\sum_{j=1}^{i-L_i-1} (S_j-C_j),\label{eq:addis_spending_loc}\end{align}
where $(\gamma_i)_{i\in \mathbb{N}}$ is the same as in an ADDIS-Graph but non-increasing. Note that $t(i)^{\text{loc}}$ increases when $L_i$ increases and thus $\alpha_i^{\text{spend}}$ decreases. Therefore, the $\text{ADDIS-Spending}_{\text{local}}$ loses significance level due to the dependency of the $p$-values. This is the main difference to the $\text{ADDIS-Graph}_{\text{conf}}$, where we argued that the same significance level is distributed under any conflict sets, as long as there is a non-conflicting future hypothesis. In the following, we use this to define an $\text{ADDIS-Graph}_{\text{conf}}$ that uniformly improves the $\text{ADDIS-Spending}_{\text{local}}$. For this, we first show how the $\text{ADDIS-Spending}_{\text{local}}$ can be written as something similar to an ADDIS-Graph.

 \begin{lemma}\label{lemma:uniform_improvement}
 Let $(\gamma_i)_{i\in \mathbb{N}}$ be non-increasing and define 
 
 \begin{align*}\alpha_i^{\text{ind}} & \coloneqq (\tau_i-\lambda_i)\left(\alpha \gamma_i + \sum_{j=1}^{i-1} g_{j,i} (C_j-S_j+1)  \frac{\alpha_j^{\text{ind}}}{\tau_j-\lambda_j}\right)\\
 \alpha_i^{\text{loc}} &\coloneqq (\tau_i-\lambda_i)\left(\alpha \gamma_i + \sum_{j=1}^{i-L_i-1} g_{j,i} (C_j-S_j+1)  \frac{\alpha_j^{\text{ind}}}{\tau_j-\lambda_j} \right) .
 \end{align*} 
 
 Then $\alpha_i^{\text{ind}}$ is equivalent to an ADDIS-Spending under independence ($L_i=0 \text{ } \forall i\in \mathbb{N}$) and $\alpha_i^{\text{loc}}=\alpha_i^{\text{spend}}$ for all $i\in \mathbb{N}$,  if $g_{j,i}=\frac{\gamma_{t(j)+i-j-1}-\gamma_{t(j)+i-j}}{\gamma_{t(j)}}$, $i>j$, where $t(j)=1+\sum_{k<j} (S_k-C_k)$.
 \end{lemma}

% \begin{lemma}\label{lemma:uniform_improvement}
% Let $(\gamma_i)_{i\in \mathbb{N}}$ be non-increasing and define \begin{align*}\alpha_i^k & \coloneqq (\tau_i-\lambda_i)\left(\alpha \gamma_i + \sum_{j=1}^{i-1} g_{j,i}^k (C_j-S_j+1)  \frac{\alpha_j^k}{\tau_j-\lambda_j}\mathbbm{1}\{i\notin \mathcal{X}_k\}\right),
% \end{align*} 
% where $g_{j,i}^k=\frac{\gamma_{t(j)^k+i-j-1}-\gamma_{t(j)^k+i-j}}{\gamma_{t(j)^k}}$, $i>j$, with $t(j)^k=1+\sum_{l<j} \mathbbm{1}\{l\in \mathcal{X}_k\}+\sum_{l<j} (S_k-C_k) \mathbbm{1}\{l\notin \mathcal{X}_k\}$. Then $\alpha_i^i=\alpha_i^{\text{spend}}$ for all $i\in \mathbb{N}$.
% \end{lemma}

 Lemma \ref{lemma:uniform_improvement} shows that $\text{ADDIS-Spending}_{\text{local}}$ can be represented as an ADDIS-Graph, in which significance level is distributed to all future hypotheses, even to the dependent ones, but only levels that come from independent hypotheses are used to test a hypothesis. It is intuitive, that it is more efficient to directly distribute significance level only to independent hypotheses. To see this, consider an example with $L_2=1$ and $L_3=0$. Then $\alpha_1^{\text{loc}}=(\tau_1-\lambda_1)\alpha \gamma_1$, $\alpha_2^{\text{loc}}=(\tau_2-\lambda_2)\alpha \gamma_2$ and $\alpha_3^{\text{loc}}=(\tau_3-\lambda_3)(\alpha \gamma_3+U_1 \alpha \gamma_1 g_{1,3}  + U_2 \alpha \gamma_2g_{2,3}+U_2U_1g_{1,2}g_{2,3} \alpha\gamma_1)$, where $U_i=C_i-S_i+1$. Now, if we replace $g_{1,2}$ by $g_{1,2}^*=0$ and $g_{1,i}$ by $g_{1,i}^*=g_{1,i}+g_{1,2}g_{2,i}$, $i>2$, we would distribute the same amount of level as before. However, while $\alpha_1^{\text{loc}}$ and $\alpha_2^{\text{loc}}$ remain the same, we obtain $\alpha_3^{\text{loc}}=(\tau_3-\lambda_3)(\alpha \gamma_3+U_1 \alpha \gamma_1 (g_{1,3}+g_{1,2}g_{2,3})  + U_2 \alpha \gamma_2g_{2,3})$, which is larger than before since $U_2\leq 1$. In the proof of the following theorem, we introduce a general algorithm exploiting this to uniformly improve the $\text{ADDIS-Spending}_{\text{local}}$.

% Lemma \ref{lemma:uniform_improvement} shows that $\alpha_i^{\text{spend}}$ can be represented as an ADDIS-Graph, in which significance level is distributed to all future hypotheses that do not conflict with $H_i$. It is intuitive, that this is inefficient. To see this, consider an example with $\mathcal{X}_2=\{1\}$ and $\mathcal{X}_3=\emptyset$. Then $\alpha_1^{1}=\alpha \gamma_1$, $\alpha_2^{2}=\alpha \gamma_2$ and $\alpha_3^3=\alpha \gamma_3+U_1 \alpha \gamma_1 g_{1,3}^3  + U_2 \alpha (g_{2,3}^3\gamma_2+U_3g_{1,2}^3g_{2,3}^3 \gamma_1)$. Now, if we define $\tilde{g}_{1,2}^i=0$ for all 

% In the proof of the following theorem, we introduce a choice of adjusted weights $(g_{j,i}^*)_{j\in \mathbb{N},i>j}$ for $\text{ADDIS-Graph}_{\text{conf}}$ that exploits this to uniformly improve the $\text{ADDIS-Spending}_{\text{conf}}$.

\begin{proposition}\label{theo:uniform_improvement}
    Let $(\gamma_i)_{i\in \mathbb{N}}$ be non-increasing and define $\mathcal{X}_i=\{i-1,\ldots, i-L_i\}$. Then there exists a choice of weights $(g_{j,i}^*)_{j\in \mathbb{N}, i>j}$ such that the $\text{ADDIS-Graph}_{\text{conf}}$  uniformly improves the $\text{ADDIS-Spending}_{\text{local}}$. We denote this procedure as $\text{ADDIS-Graph}_{\text{conf-u}}$.
\end{proposition}

\section{Extensions of the ADDIS-Graph}\label{sec:extensions}

\subsection{Control of the PFER}\label{sec:PFER}
Tian and Ramdas (2021) \citep{TR} showed that procedures satisfying \eqref{eq:cond_addis_principle} even control the more conservative per-family wise error rate (PFER) defined by 
\begin{align}
\text{PFER}\coloneqq \mathbb{E}[v],  \label{eq:PFER}
\end{align}
where $v$ is the number of false rejections. Hence, it is not an actual extension of the ADDIS-Graph, but an immediate consequence of the result by Tian and Ramdas (2021) \citep{TR}, that the ADDIS-Graph also controls the PFER. It follows by $\mathbbm{1}\{v>0\}\leq v$ that control of the PFER implies strong control of the FWER (see \eqref{eq:FWER_PFER} for another explanation).

While we focused on the FWER since it is the more common error rate in practice, the PFER still offers good interpretability and for this reason there are also applications where the PFER is desirable. For example, in a platform trial many treatment groups are compared to the same control group. Now, if the control group performed badly, meaning worse outcomes were observed than usual, there is a danger of deeming many treatments as efficient even though they are not. The FWER does not protect against this case, as it only ensures that the probability of committing any type I error is small. However, if we are in the case of a type I error, it has no guarantee about the number of type I errors. This can be resolved by controlling the PFER and is therefore automatically provided by the ADDIS-Graph. This example is not intended to question the appropriateness of the FWER for platform trials, but only to show that the control of the PFER provides additional sensible control.

\subsection{Comparison to closed ADDIS-Spending and the closed ADDIS-Graph}\label{sec:closed_ADDIS}

Fischer et al. (2024) \citep{fischer2022online} introduced another improvement of the ADDIS-Spending under local dependence based on their online closure principle. The \textit{closed ADDIS-Spending} tests each individual hypothesis at the level 
\begin{align}\alpha_i^{\text{c-spend}} = \alpha (\tau_i-\lambda_i) \gamma_{t(i)^{\text{c-loc}}}, \quad \text{where } t(i)^{\text{c-loc}}=1+ \sum_{j=i-L_i}^{i-1} (1-R_j) +\sum_{j=1}^{i-L_i-1} (S_j-\max\{R_j,C_j\}),\label{eq:closed_addis_spending}\end{align}
where $R_j=\mathbbm{1}\{P_j\leq \alpha_j^{\text{c-spend}}\}$. Note that this is a uniform improvement of the ADDIS-Spending under local dependence \eqref{eq:addis_spending_loc}, because $\sum_{j=i-L_i}^{i-1} (1-R_j)\leq L_i$ and $\max\{R_j,C_j\}\geq C_j$. We usually choose $\lambda_i\geq \alpha_i$ (see \citep{TR}) and also assume this in the following argumentation such that the latter inequality becomes an equation and the only improvement comes from the former inequality. Hence, the difference compared  to the ADDIS-Spending under local dependence is that even if $P_j$ and $P_i$, $i>j$, depend on each other, the level $\alpha_i$ is allowed to be adjusted to the information whether $P_j\leq \alpha_j$. Thus, if a hypothesis $H_j$ is rejected, the closed ADDIS-Spending no longer loses significance level due to the local dependence, however, it still does if $\alpha_j < P_j\leq \lambda_j$ or $P_j>\tau_j$. 

Hence, one difference to the uniform improvement obtained by the $\text{ADDIS-Graph}_{\text{conf}}$ introduced in Section \ref{sec:ADDIS-Spending} is that the closed ADDIS-Spending only improves the ADDIS-Spending in case of a rejection, while the $\text{ADDIS-Graph}_{\text{conf}}$ even improves the ADDIS-Spending if $\alpha_j < P_j\leq \lambda_j$ or $P_j>\tau_j$ (and local dependence is present). This suggests that the uniform improvement obtained by the $\text{ADDIS-Graph}_{\text{conf}}$ is stronger, which is verified by simulations (see Appendix \ref{sec:sim_appendix}). Another difference between the $\text{ADDIS-Graph}_{\text{conf}}$ and closed ADDIS-Spending is the construction. While the closed ADDIS-Spending is based on applying the (online) closure principle \citep{fischer2022online, MPG} directly to the ADDIS-Spending, the $\text{ADDIS-Graph}_{\text{conf}}$ defines another way to exploit condition \eqref{eq:cond_addis_principle}. The latter has the advantage that it additionally provides control of the PFER as discussed in Section \ref{sec:PFER}, while the closure principle only guarantees FWER control. Furthermore, the closed ADDIS-Spending cannot be formulated for general conflict sets and the $\text{ADDIS-Graph}_{\text{conf}}$ provides a higher flexibility in general.

Nevertheless, this poses the question of whether the $\text{ADDIS-Graph}_{\text{conf}}$ can also be improved by the (online) closure principle under local dependence. Similarly as for the ADDIS-Spending, one can construct a closed $\text{ADDIS-Graph}_{\text{conf}}$ under local dependence that tests the individual hypothesis $H_i$ at the level
\begin{align}
\alpha_i^{\text{c-graph}} = (\tau_i-\lambda_i)\left(\alpha \gamma_i + \sum_{j=i-L_i}^{i-1} g_{j,i} R_j  \frac{\alpha_j}{\tau_j-\lambda_j} + \sum_{j=1}^{i-L_i-1} g_{j,i} (\max\{R_j,C_j\}-S_j+1)  \frac{\alpha_j}{\tau_j-\lambda_j}\right), \label{eq:closed_graph}
\end{align}
where $R_j=\mathbbm{1}\{P_j\leq \alpha_j^{\text{c-graph}}\}$ and we usually assume that $\max\{R_j,C_j\}=C_j$ as above. 
We provide the derivation of this closed $\text{ADDIS-Graph}_{\text{conf}}$ in Appendix \ref{sec:closed_ADDIS_graph}. The closed $\text{ADDIS-Graph}_{\text{conf}}$ allows to distribute significance level to dependent hypotheses in the case of a rejection and to independent hypotheses if $P_i\leq \lambda_i$ or $P_i>\tau_i$. Since the $\text{ADDIS-Graph}_{\text{conf}}$ can only distribute significance level to independent hypotheses (in case of $P_i\leq \lambda_i$ or $P_i>\tau_i$), the closed $\text{ADDIS-Graph}_{\text{conf}}$ could be seen as a uniform improvement of it. However, $g_{j,i}$ is only allowed to depend on information that is independent of $P_i$. Hence, if $P_j$ and $P_i$ depend on each other, $g_{j,i}$ must be fixed before knowing the true value of $P_j$ and we must decide whether we want to distribute some of the significance level $\alpha_j^{\text{c-graph}}$ to dependent hypotheses (in case of a rejection) before knowing whether $H_j$ will be rejected. Since it is much more likely that $P_j\leq \lambda_j$ or $P_j>\tau_j$ than $P_j\leq \alpha_i^{\text{c-graph}}$, one usually chooses $g_{j,i}=0$ for all $i>j$ with $i-L_i\leq j$ in order to obtain a high power. In this case, the closed $\text{ADDIS-Graph}_{\text{conf}}$ reduces to the usual $\text{ADDIS-Graph}_{\text{conf}}$ under conflict sets. The only situation in which we believe that the closed $\text{ADDIS-Graph}_{\text{conf}}$ would provide a real uniform improvement of the $\text{ADDIS-Graph}_{\text{conf}}$ is when we have the information that all future hypotheses are dependent on the current $P_j$. In this case, it would be best to choose $\lambda_j=0$ and $\tau_j=1$ such that the closed $\text{ADDIS-Graph}_{\text{conf}}$ reduces to the online version \citep{TR, fischer2022online} of the classical graphical approach by Bretz et al. (2009) \citep{BWBP}, while the $\text{ADDIS-Graph}_{\text{conf}}$ becomes the more conservative online Bonferroni adjustment \citep{FS, TR}.

\subsection{Exploiting the joint distribution of the p-values}

In Section \ref{sec:PFER}, we have noted that the ADDIS-Graph even controls the more conservative error rate PFER and in Section \ref{sec:closed_ADDIS} we have shown that the (online) closure principle does not give a direct improvement of the ADDIS-Graph. Hence, the question is how the gap between PFER and FWER control can be used to improve the ADDIS-Graph further. For this, note that the connection between the PFER and the FWER can also be explained by the Bonferroni inequality
\begin{align}
\text{FWER}=\mathbb{P}\left(\bigcup_{i\in I_0} \{P_i\leq \alpha_i\}\right)\leq \sum_{i\in I_0} \mathbb{P}(P_i\leq \alpha_i)= \sum_{i\in I_0} \mathbb{E}[\mathbbm{1}\{P_i\leq \alpha_i\}]=\mathbb{E}[v]=\text{PFER}. \label{eq:FWER_PFER}
\end{align}
It is known that the Bonferroni inequality leads to conservative procedures as it makes worst case assumptions about the joint distribution of the p-values. Hence, one approach to improve the ADDIS-Graph would be to incorporate information about the joint distribution of the p-values. Note that such information is not always available. However, for example in a platform trial, the entire dependency between the p-values comes from the shared control data and therefore can be determined if the number of observations shared is available. In Appendix \ref{sec:correlation_structure}, we demonstrate how such information about the correlation structure can be included in the $\text{ADDIS-Graph}_{\text{conf}}$ and compare it via simulations to the usual $\text{ADDIS-Graph}_{\text{conf}}$ in Appendix \ref{sec:sim_appendix}. However, note that the proposed method only works if $\tau_i=1$ and the local dependence structure is given by batches. 

\subsection{ADDIS-Graphs for FDR control}
While being the norm in validation studies \citep{wason2014correcting}, FWER control can be too conservative for certain applications, particularly if the number of hypotheses is large. Less conservative error rates often considered in the (online) multiple testing literature are the false discovery rate (FDR) \citep{BH} and the modified FDR (mFDR), where 
\begin{align}
    \text{FDR(}i\text{)}\coloneqq \mathbb{E}\left(\frac{|V(i)|}{|R(i)|\lor 1}\right) \qquad \text{mFDR(}i\text{)}\coloneqq \frac{\mathbb{E}(|V(i)|)}{\mathbb{E}(|R(i)|\lor 1)} \qquad (i\in \mathbb{N}).
\end{align}
The goal is to control $\text{FDR(}i\text{)}$ or $\text{mFDR(}i\text{)}$ at each time $i\in \mathbb{N}$. While the FDR is the most common error rate in large scale multiple testing, mFDR control is often considered in online multiple testing since it is easier to prove and usually requires fewer assumptions \citep{FS, JM, ZRJ}. Note that there is a strong connection to the PFER, which is basically the numerator of the mFDR. Hence, it is not surprising that PFER procedures can be extended comparatively easy to mFDR or even FDR control. Tian and Ramdas (2019) \citep{TR2} have provided a condition similar to \eqref{eq:cond_addis_principle} that allows to construct powerful ADDIS procedures with control of the FDR or mFDR under different assumptions. In Appendix \ref{sec:FDR} we introduce an improved version of the ADDIS-Graph that provably satisfies the condition by Tian and Ramdas (2019) \citep{TR2}. Furthermore, we argue that the resulting FDR-ADDIS-Graph outperforms the existing ADDIS methods when conflicts are present.

\section{Simulations}\label{sec:sim}
% We investigate the power and error control of the proposed ADDIS-Graphs by means of simulations. In Subsection \ref{sec:sim_FWER}, we compare the  $\text{ADDIS-Spending}_{\text{local}}$\citep{TR} with its uniform improvement $\text{ADDIS-Graph}_{\text{conf-u}}$. 
% Afterwards, 
% % we examine the effect of incorporating information about the correlation structure into Adaptive procedures in a local dependence setup (Subsection \ref{sec:sim_corr}) and 
% in Subsection  \ref{sec:sim_FDR}, we compare the $\text{FDR-ADDIS-Graph}_{\text{conf}}$ with the $\text{ADDIS}^*$ algorithm\citep{TR2} in an asynchronous testing setup.

% \subsection{Comparison of ADDIS-Graph and ADDIS-Spending under local dependence\label{sec:sim_FWER}}

In this section, we compare the power and FWER control of the $\text{ADDIS-Graph}_{\text{conf-u}}$ and the $\text{ADDIS-Spending}_{\text{local}}$ under local dependence using simulated data. Simulations for the extended ADDIS-Graphs introduced in Section \ref{sec:extensions} can be found in Appendix \ref{sec:sim_appendix}.

We consider $n=100$ hypotheses to be tested, whose corresponding $p$-values $(P_i)_{i\in \{1,\ldots,n\}}$ follow a batch dependence structure $B_{1},\ldots,B_{n/b}$ with the same batch-size $b\in \{1,5,10,20\}$ for every batch. That means $B_j=\{P_{(j-1)b+1}, \ldots,P_{jb}\}$, $j\in \{1,\ldots, n/b\}$, and all p-values within a batch $B_j$ depend on each other, while p-values from different batches are independent. Let $X_{(j-1)b+1:bj}=(X_{(j-1)b+1},\ldots,X_{bj})^T {\sim} N_b(\mu,\Sigma)$ be $b$-dimensional $i.i.d$ random vectors, where  $\mu= (0,\ldots,0)^T\in \mathbb{R}^{b}$ and $\Sigma=(\sigma_{ik})_{i,k=1,\ldots,b}\in \mathbb{R}^{b\times b}$ with $\sigma_{ii}=1$ and $\sigma_{ik}=\rho\in (0,1)$ for all $i\in \{1,\ldots,b\}$ and $k\neq i$. For each $H_i$, $i\in \{1,\ldots,n\}$, we test the null hypothesis $H_i: \mu_i \leq 0$ with $\mu_i=\mathbb{E}[Z_i]$, where
 $Z_i=X_i+3$ with probability $\pi_A\in (0,1)$ and $Z_i=X_i+\mu_N$, $\mu_N<0$, otherwise.
 Since the test statistics follow a standard Gaussian distribution under the null hypothesis, a  z-test can be used. The parameter $\pi_A$ can be interpreted as the probability of a hypothesis being false and $\mu_N$ as the conservativeness of null $p$-values \citep{TR}. 
 
 In this subsection, we use an overall level $\alpha=0.2$ and estimate the FWER and power of the $\text{ADDIS-Graph}_{\text{conf-u}}$ and $\text{ADDIS-Spending}_{\text{local}}$ \citep{TR} by averaging over $1000$ independent trials. Thereby, the proportion of rejected hypotheses among the false hypotheses is used as empirical power. We set $\mu_N=-0.5$ and $\rho=0.5$ in all simulations within this section, thus obtaining slightly conservative null $p$-values with positive correlation within each batch. Since both procedures are based on the same ADDIS principle and therefore exploit the conservativeness of null $p$-values in the same manner, no more parameter configurations are necessary. As recommended \cite{TR}, we choose $\tau_i=0.8$ and $\lambda_i=\alpha \tau_i=0.16$ for all $i\in \mathbb{N}$. The rows in Figure \ref{fig:plot_fwer_logq} vary with respect to the chosen $(\gamma_i)_{i\in \mathbb{N}}$, as the procedures are sensitive to it. In the top row, we use $\gamma_i\propto 1/\left((i+1)\log(i+1)^2\right)$, in the middle row $\gamma_i\propto 1/i^{1.6}$ and in the bottom row $\gamma_i=6/(\pi^2 i^2)$.

\begin{figure}[htbp]
	\begin{center}
			\centering
		\includegraphics[width=19.5cm,height=6.5cm,keepaspectratio]{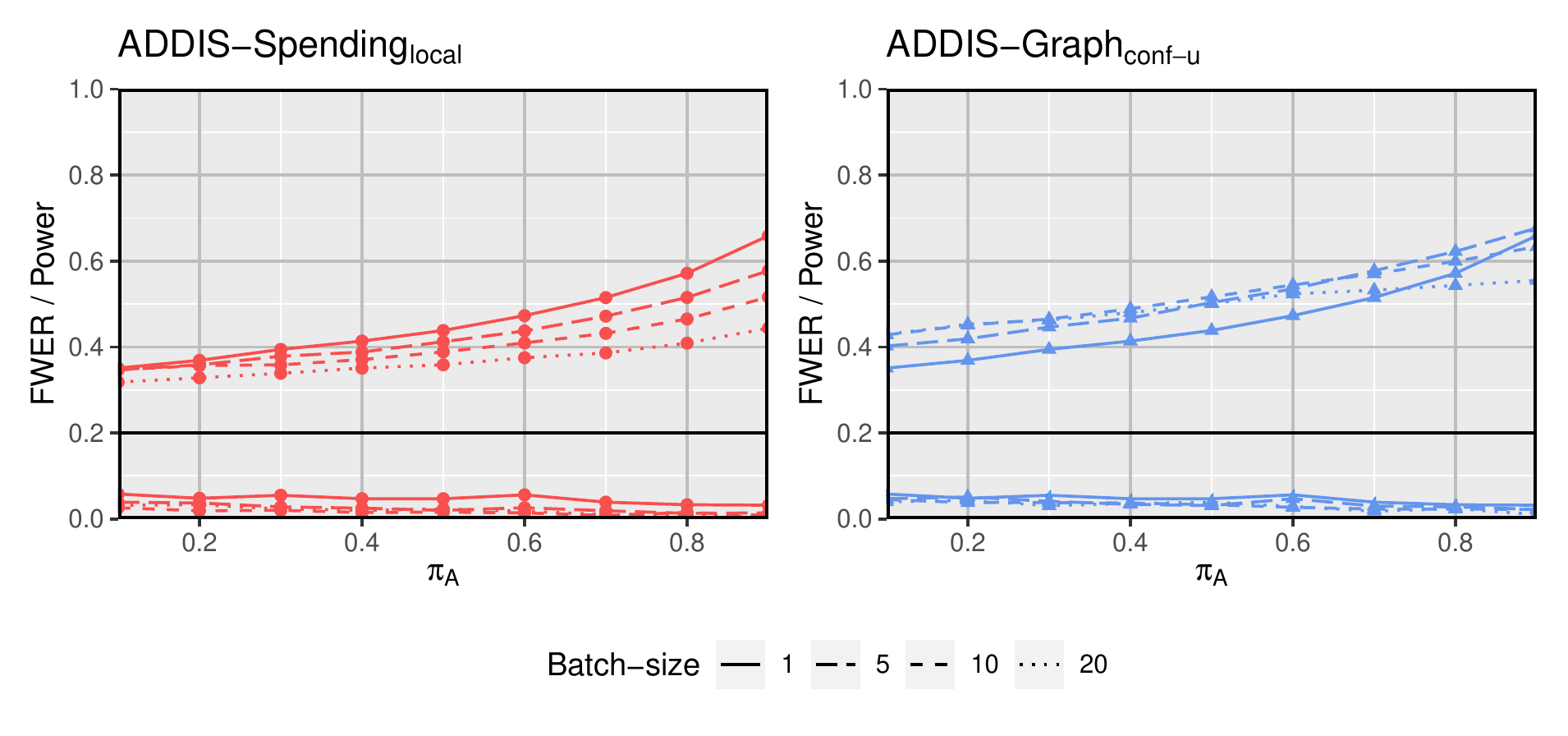}
  \includegraphics[width=19.5cm,height=6.5cm,keepaspectratio]{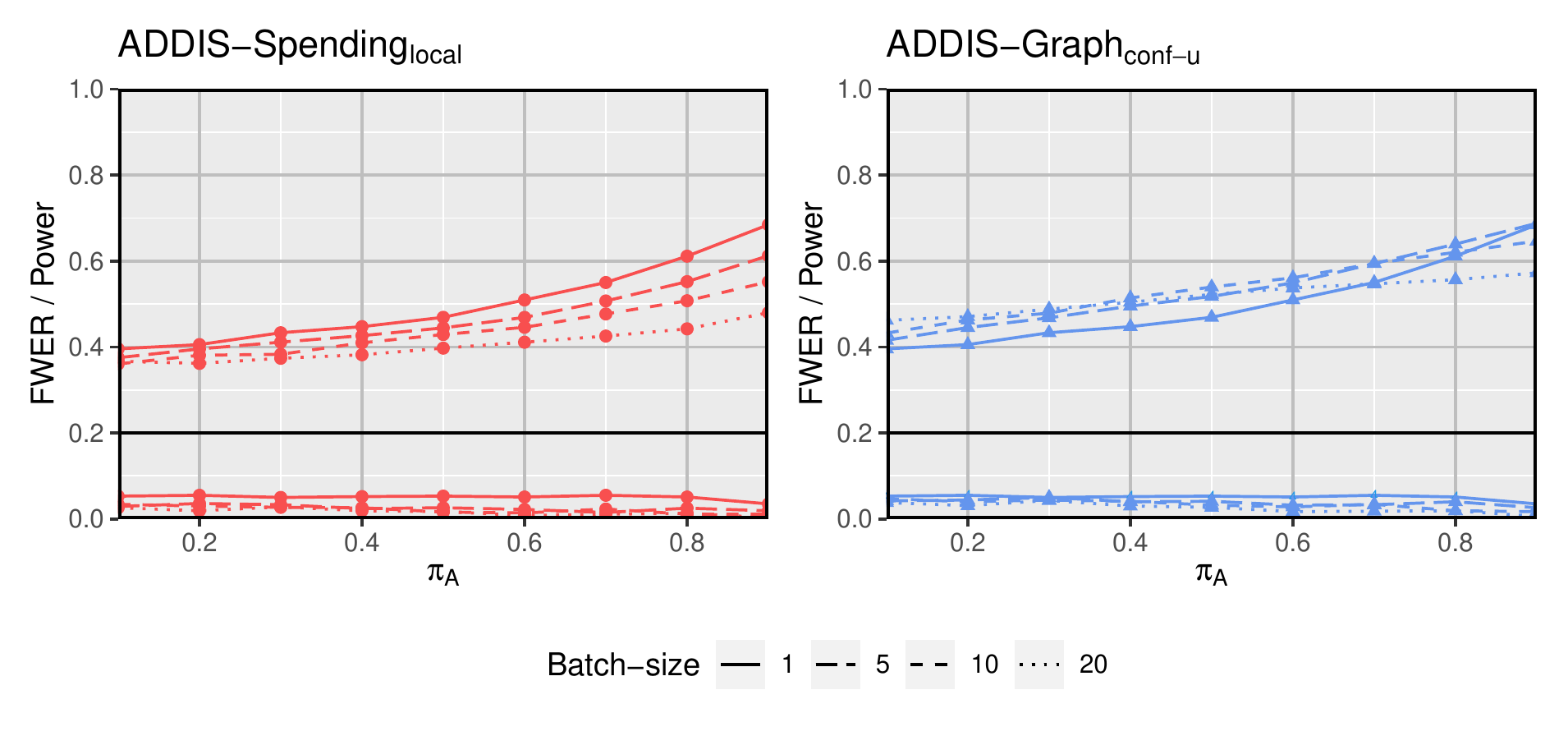}
  \includegraphics[width=19.5cm,height=6.5cm,keepaspectratio]{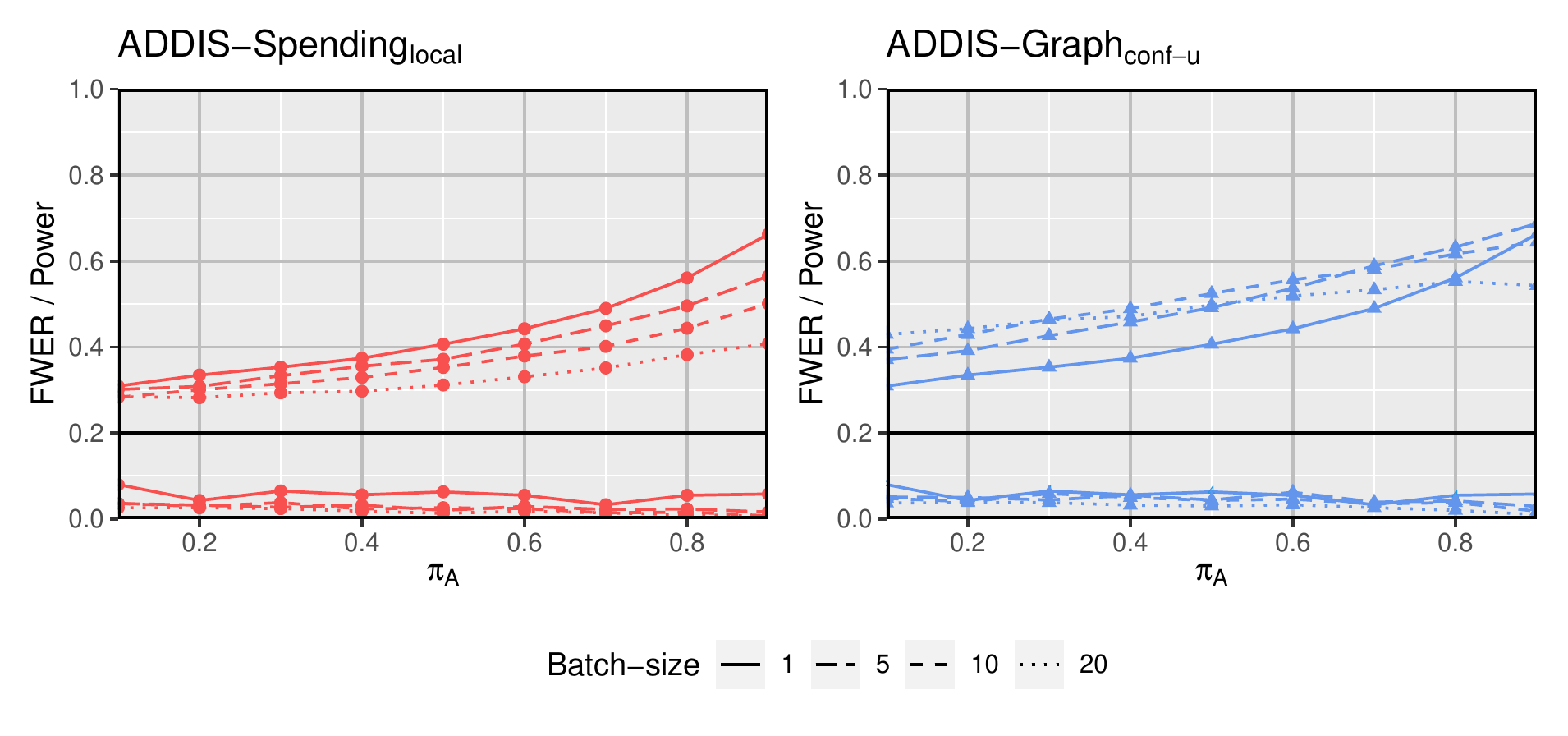}
	\end{center}
	\caption{Comparison of $\text{ADDIS-Spending}_{\text{local}}$ and $\text{ADDIS-Graph}_{\text{conf-u}}$ in terms of power and FWER for different batch-sizes and proportions of false null hypotheses ($\pi_A$). Lines above the overall level $\alpha=0.2$ correspond to power and lines below to FWER. The $p$-values were generated as described in the text with parameters $\mu_N=-0.5$ and $\rho=0.5$. Both procedures were applied with parameters $\tau_i=0.8$ and $\lambda_i=0.16$. In the top row $\gamma_i\propto 1/\left((i+1)\log(i+1)^2\right)$, in the middle row $\gamma_i\propto 1/i^{1.6}$ and in the bottom row $\gamma_i=6/(\pi^2 i^2)$. Under independence of the $p$-values both procedures coincide. However, the $\text{ADDIS-Spending}_{\text{local}}$ loses power when the $p$-values become locally dependent, while the $\text{ADDIS-Graph}_{\text{conf-u}}$ offers a similar or even larger power as under independence. \label{fig:plot_fwer_logq}}
\end{figure}

% \begin{figure}[htbp]
% 	\begin{center}
% 			\centering
% 		\includegraphics[width=19.5cm,height=6.5cm,keepaspectratio]{Plot_FWER_ADDIS_Graph_local-u_q^1,6.pdf}
% 	\end{center}
% 	\caption{Comparison of $\text{ADDIS-Spending}_{\text{local}}$ and $\text{ADDIS-Graph}_{\text{local-u}}$ in terms of power and FWER for different batch-sizes and proportions of false null hypotheses ($\pi_A$). Lines above the overall level $\alpha=0.2$ correspond to power and lines below to FWER. The $p$-values were generated as described in the text with parameters $\mu_N=-0.5$ and $\rho=0.5$. Both procedures were applied with parameters $\tau_i=0.8$, $\lambda_i=0.16$ and $\gamma_i\propto 1/i^{1.6}$. \label{fig:plot_fwer_q^1.6}}
% \end{figure}

% \begin{figure}[htbp]
% 	\begin{center}
% 			\centering
% 		\includegraphics[width=19.5cm,height=6.5cm,keepaspectratio]{Plot_FWER_ADDIS_Graph_local-u_q^2.pdf}
% 	\end{center}
% 	\caption{Comparison of $\text{ADDIS-Spending}_{\text{local}}$ and $\text{ADDIS-Graph}_{\text{local-u}}$ in terms of power and FWER for different batch-sizes and proportions of false null hypotheses ($\pi_A$). Lines above the overall level $\alpha=0.2$ correspond to power and lines below to FWER. The $p$-values were generated as described in the text with parameters $\mu_N=-0.5$ and $\rho=0.5$. Both procedures were applied with parameters $\tau_i=0.8$, $\lambda_i=0.16$ and $\gamma_i=6/(\pi^2 i^2)$. \label{fig:plot_fwer_q^2}}
% \end{figure}

As shown in Lemma \ref{lemma:uniform_improvement}, the procedures are equivalent under independence of the $p$-values. However, when the $p$-values become locally dependent, the power of the $\text{ADDIS-Spending}_{\text{local}}$ decreases systematically in all cases, while the power of the $\text{ADDIS-Graph}_{\text{conf-u}}$ even increases in most cases. To understand why the power might increase under local dependence, note that 
the larger the batch-size, the further into the future the significance level is distributed by the weights $(g^{*}_{j,i})_{i=j+1}^{\infty}$ (see Algorithm \ref{alg:1} in the Appendix). 
This can lead to a more evenly distribution of the significance level under a larger batch-size, which results in a higher power. 
However, if $(\gamma_i)_{i\in \mathbb{N}}$ decreases slowly and the batch-size is large, a lot of the significance level is distributed to hypotheses in the far future. Since the testing process is finite in this case, these hypotheses may never be tested, which is why power is lost when $\pi_A$ is large.

% In these simulations $\gamma_i\propto 1/\left((i+1)\log(i+1)^2\right)$ (Figure \ref{fig:plot_fwer_logq}) decreases slowest and $\gamma_i=6/(\pi^2 i^2)$ (Figure \ref{fig:plot_fwer_q^2}) decreases fastest. If $(\gamma_i)_{i\in \mathbb{N}}$ decreases slow and the batch-size is large, $\text{ADDIS-Graph}_{\text{local}}$ distributes a lot of significance level to hypotheses in the far future. However, since the testing process is finite in this case, these hypotheses may never be tested, which leads to a power loss. On the other hand, if $(\gamma_i)_{i\in \mathbb{N}}$ decreases fast, $\text{ADDIS-Graph}_{\text{local}}$ allocates the individual significance levels more evenly under a larger batch-size, which results in a higher power. Thus, in order to obtain the optimal power for each batch-size, one could choose a faster decreasing $(\gamma_i)_{i\in \mathbb{N}}$ the larger the batch-size.

\section{Application to RECOVERY trial}\label{sec:recovery_trial}
In this section, we illustrate the usage of the ADDIS-Graph by applying it on a real ongoing platform trial. The Randomised Evaluation of COVID-19 Therapy (RECOVERY) trial was launched in 2020 and evaluates treatments for severe COVID-19 diseases against a standard of care. Up to this date, twelve treatments have already been tested, while a thirteenth treatment is currently recruiting \cite{sandercock2022experiences}. 
The $p$-values are reported at the website \url{https://www.recoverytrial.net/}. The platform trial structure is illustrated in Figure \ref{fig:recovery} and was copied from a publication by the data monitoring committee \cite{sandercock2022experiences}. As exemplified in Figure \ref{fig:platform_trial}, overlapping hypotheses share some control data and are therefore conflicting. The ADDIS-Graph can be used to adapt to these conflict sets. For example, the treatment arms $T_1$, $T_2$ and $T_3$ only distribute significance level to the treatment $T_7$ and onwards, $T_4$ and $T_6$ distribute level to the treatment $T_{10}$ and onwards, and $T_5$ to treatment $T_{11}$ and onwards.  

\begin{figure}[htbp]
	\begin{center}
			\centering
		\includegraphics[width=21cm,height=7cm,keepaspectratio]{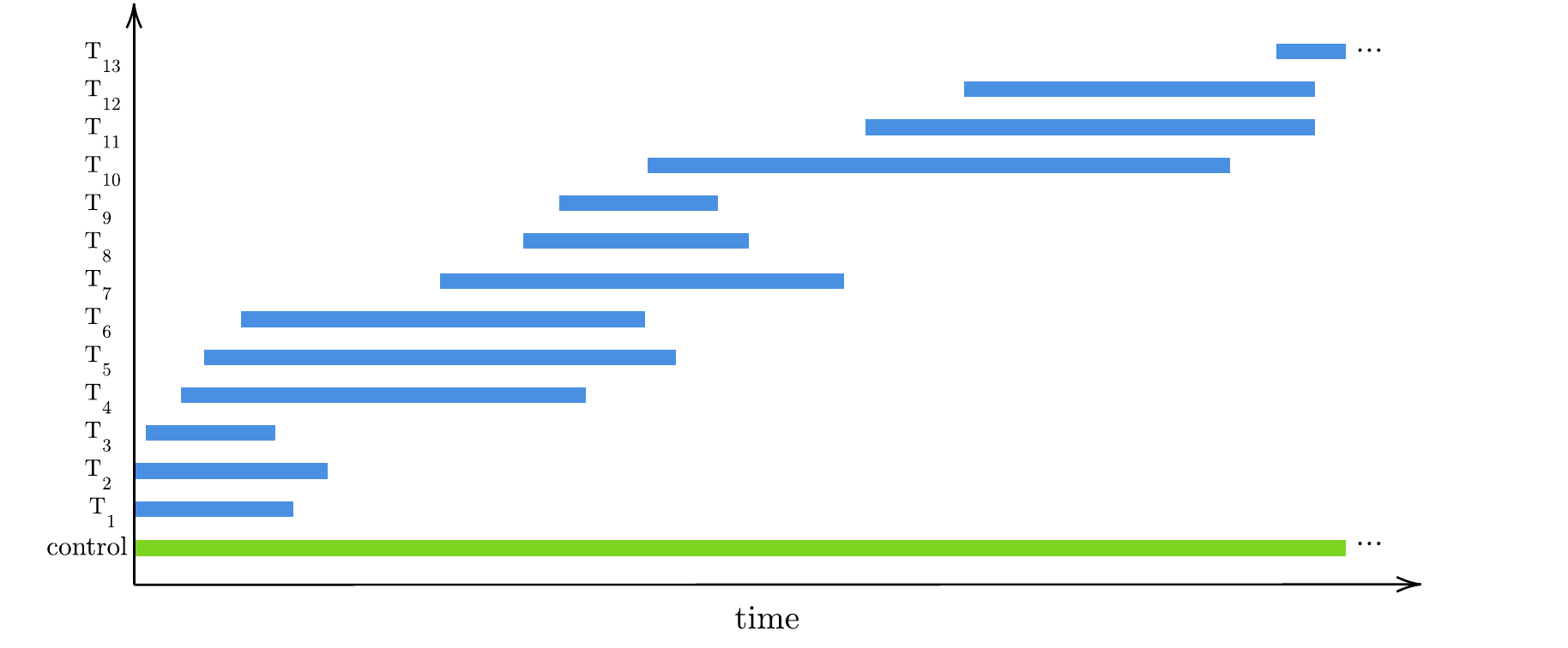}
	\end{center}
	\caption{Overlapping structure of the RECOVERY trial \citep{sandercock2022experiences}. \label{fig:recovery}}
\end{figure} 

We compare the obtained rejections and the remaining significance level for future testing when applying the $\text{ADDIS-Spending}_{\text{local}}$ and the $\text{ADDIS-Graph}_{\text{conf-u}}$. In addition, we provide the results for an uncorrected procedure which tests each hypothesis at full level $\alpha=0.05$ as a reference. Similarly as done by Fischer et al. (2024) \citep{fischer2022online}, we set $\gamma_i=q^i \frac{1-q}{q}$ for $q\in \{0.6,0.7,0.8\}$. Note that the larger the $q$, the slower $(\gamma_i)_{i\in \mathbb{N}}$ converges to $0$. We set $\tau_i=0.8$ and $\lambda_i=0.3$ for the ADDIS procedures.

The results are summarised in Table \ref{table:results_recovery}. The level for future hypotheses was calculated as the sum of future significance levels if a Bonferroni adjustment would be applied, meaning if one sets $\tau_i=1$ and $\lambda_i=0$, $i>12$. The results show that the $\text{ADDIS-Graph}_{\text{conf-u}}$ was able to reject one more hypothesis compared to $\text{ADDIS-Spending}_{\text{local}}$ in case of $q=0.6$. In addition, $\text{ADDIS-Graph}_{\text{conf-u}}$ leaves considerably more level for future hypotheses such that it is much more likely to obtain additional rejections in the future than with the $\text{ADDIS-Spending}_{\text{local}}$. Furthermore, the $\text{ADDIS-Graph}_{\text{conf-u}}$ appears to be more robust against the choice of $q$. 

% Furthermore, the $\text{Adaptive-Graph}_{\text{corr}}$ leaves the most level for the testing of future hypotheses. Note that in case of $q=0.9$ only $20\%$ of the available significance level was actually used, indicating that a lot more hypotheses could be tested at a similar level. In total, the $\text{ADDIS-Graph}_{\text{local-u}}$ and  $\text{Adaptive-Graph}_{\text{corr}}$ appear to be most robust against the choice of $q$ and leave the most significance level for future hypotheses such that the difference between the number of rejections will possibly be even larger when further hypotheses are tested.

\begin{table}[h!]
\centering
\caption{Number of rejections and level for future hypotheses obtained by different procedures applied on the RECOVERY trial.}
\begin{tabular}[h]{|c|c|c|c|c|c|c|}
\hline
\multirow{2}{*}{\textbf{Procedure}}  & \multicolumn{3}{c|}{\textbf{Number of rejections}} &  \multicolumn{3}{c|}{\textbf{Level for future hypotheses}} \\
  & $q=0.6$ & $q=0.7$ & $q=0.8$ & $q=0.6$ & $q=0.7$ & $q=0.8$  \\
  \hline
  $\text{ADDIS-Spending}_{\text{local}}$ & $2$ & $3$ & $3$ & $0.0039$ & $0.0084$ & $0.0164$ \\ 
  $\text{ADDIS-Graph}_{\text{conf-u}}$ & $3$ & $3$ & $3$ & $0.0256$ & $0.0246$ & $0.0263$ \\
  Uncorrected & $5$ & $5$ & $5$ & $\infty$ & $\infty$ & $\infty$ \\
  \hline
\end{tabular}
\label{table:results_recovery}
\end{table}

\section{Discussion}

% In this work, we presented a graphical approach to exploit the ADDIS principles for FWER \citep{TR} and FDR \citep{TR2} control. We started with the construction of an FWER controlling ADDIS-Graph. 
% This proposal enhances the interpretability of the ADDIS-Spending and also enlarges the family of procedures that it includes, as we show that by means of the ADDIS-Graph all procedures that satisfy the ADDIS principle for FWER control can be obtained. The ADDIS-Graph can easily be adapted to a local dependence structure and an asynchronous testing setup, leading to a uniform improvement over the ADDIS-Spending. Furthermore, we constructed a new adaptive online testing principle, which allows to exploit correlation structures in a local dependence setting while adapting to the proportion of false null hypotheses. Moreover, we extended the ADDIS-Graph to the FDR control setting resulting in an FDR-ADDIS-Graph. It has similar advantages as the ADDIS-Graph and we showed by means of simulations that it is superior to the currently used ADDIS method with FDR control, the $\text{ADDIS}^*$ algorithm. 

In their review paper, Robertson et al. \citep{robertson2022online} named the construction of online procedures for a small number of hypotheses 
and with known correlation structure, especially with respect to platform trials, as one of the future directions in online multiple testing. In addition, they claimed that the individual significance levels assigned by asynchronous online procedures are more conservative. In this paper, we constructed ADDIS-Graphs that, due to their graphical structure, perfectly adapt to such complex trial designs (see e.g. Figure \ref{fig:platform_trial}). We demonstrated that the ADDIS-Graphs lead to power improvements over the current state-of-art methods, as the level that is lost due to pessimistic assumptions because of local dependence or asynchronous testing, is reused at later steps, such that no significance level is lost overall. In particular, we showed that the ADDIS-Graph for FWER control uniformly improves the ADDIS-Spending under local dependence \citep{TR}. Due to their graphical structure \citep{BWBP}, ADDIS-Graphs are flexible and easily comprehensible --- and therefore facilitate the planning and conduction of a trial.  For example, when the same sponsors run several treatment arms in a platform trial, they may want that significance level is only distributed between their hypotheses, which could easily be incorporated by an ADDIS-Graph, but not by the ADDIS-Spending.

% In this work, we introduced a graphical ADDIS procedure for online error control. The main advantage compared to the existing ADDIS-Spending\citep{TR} is its flexibility and easy interpretability. We show that this ADDIS-Graph contains all procedures satisfying the ADDIS condition \eqref{eq:cond_addis_principle} for FWER control by Tian and Ramdas (2021). Hence, the ADDIS-Graph turns the ADDIS condition \eqref{eq:cond_addis_principle} into a constructive procedure, which is easily comprehensible due to its graphical structure. For example, when the same sponsors run several treatment arms in a platform trial, they may want that significance level is only distributed between their hypotheses, which could easily be incorporated by an ADDIS-Graph, but not by the ADDIS-Spending\citep{TR}.

% Furthermore, when $p$-values become locally dependent, the individual significance levels of the $\text{ADDIS-Spending}_{\text{local}}$ decrease, as pessimistic assumptions about the test outcomes of dependent hypotheses are made at each step. However, we argue that the significance levels that could not be used due to dependence, can be reused at a later, independent step. In particular, we use this in Section \ref{sec:local_dependence} to show that a specific ADDIS-Graph, the $\text{ADDIS-Graph}_{\text{conf-u}}$, leads to a uniform improvement of the $\text{ADDIS-Spending}_{\text{local}}$. In general, the ADDIS-Graph perfectly adapts to all kind of conflict sets by distributing significance level only between hypotheses that are not conflicting.

We introduced several extensions of the ADDIS-Graph for FWER control. First, we showed how that the online closure principle \citep{fischer2022online} can be used to improve the ADDIS-Graph in situations where all future p-values depend on the current one. Moreover, we demonstrated how information about the joint distribution of the p-values can be incorporated to improve the procedure, which is particularly relevant for platform trials. 
Furthermore, we presented an ADDIS-Graph for FDR control.

Our proposed method for incorporating the joint distribution while adapting to the number of false hypotheses only allows to exploit  the correlation structure within batches and therefore does not unlock the full potential of the approach. We wonder whether it is possible to exhaust the entire information about the joint distribution while still adapting to the number of false hypotheses. Also, it would be interesting to additionally discard the conservative null p-values. Moreover, we argued that the FDR-ADDIS-Graph is superior to the $\text{ADDIS}^*$ algorithm for similar reasons as in the FWER case, which was also verified by simulations. However, we did not prove a uniform improvement for a specific choice of weights, which would be an interesting question for future work. Similarly, it would be interesting whether the FDR-ADDIS-Graph contains all procedures satisfying ADDIS condition for FDR control \citep{TR2}, as we only proved this for the FWER case.

\subsection*{Funding}
L. Fischer acknowledges funding by the Deutsche Forschungsgemeinschaft (DFG, German Research Foundation) – Project number 281474342/GRK2224/2.

This research was funded in whole, or in part, by the Austrian Science Fund (FWF) [ESP 442 ESPRIT-Programm]. For the purpose of open access, the author has applied a CC BY public copyright licence to any Author Accepted Manuscript version arising from this submission.

\subsection*{Conflict of interest}

The authors declare no potential conflict of interests.

% \clearpage
\subsection*{Acknowledgments}
 LF acknowledges funding by the Deutsche Forschungsgemeinschaft (DFG, German Research Foundation) – Project number 281474342/GRK2224/2. AR was funded by NSF grant DMS-2310718.

\bibliography{main}
\bibliographystyle{plainnat}

\begin{appendix}
    \section{Derivation of the closed ADDIS-Graph}\label{sec:closed_ADDIS_graph}
\vspace{0.2cm}

In this section we derive the closed $\text{ADDIS-Graph}_{\text{conf}}$ (7) as an online closed procedure \citep{fischer2022online}. We use a similar construction of the closed procedure as Fischer et al. (2024) \citep{fischer2022online} did for the Online-Graph. For this, let 
a local dependence structure $\mathcal{X}_i=\{i-1,\ldots, i-L_i\}$ be given, where $(L_i)_{i\in \mathbb{N}}$ are lags with $L_{i+1}\leq L_i+1$. Furthermore, let $(\tau_i)_{i\in \mathbb{N}}$, $(\lambda_i)_{i\in \mathbb{N}}$,
$(\gamma_i)_{i\in \mathbb{N}}$ and $(g_{j,i})_{j\in \mathbb{N},i>j}$ be sequences as in Definition 1 such that $\tau_i$, $\lambda_i$, $\gamma_i$ and $g_{j,i}$ are measurable with respect to $\mathcal{G}_{-\mathcal{X}_i}$. For each $I\subseteq \mathbb{N}$, we define an intersection test $\phi_I$ as
\begin{align*}
\phi_I&=\mathbbm{1}\left\{\exists i\in I: P_i\leq \alpha_i^I  \right\}, \\ \text{ where } \alpha_i^I&=(\tau_i-\lambda_i)\left(\alpha \gamma_i + \sum_{j\in I, j<i-L_i} g_{j,i} (C_j-S_j+1) \frac{\alpha_j^I}{\tau_j-\lambda_j} +  \sum_{j\notin I, j<i} g_{j,i} \frac{\alpha_j^{I\cup \{j\}}}{\tau_j-\lambda_j} \right).
\end{align*}
It holds that $\sum_{j\leq i,j\in I} \frac{\alpha_j^I}{\tau_j-\lambda_j}(S_j-C_j)\leq \alpha$ for all $i\in \mathbb{N}$, since for all indices that are not contained in $I$, we just shift the significance level to the future hypotheses according to the weights $(g_{j,i})_{j\in \mathbb{N}, i>j}$. With this, it follows that $\phi_I$, $I\subseteq \mathbb{N}$, is an $\alpha$-level intersection test, meaning $\mathbb{P}_{H_I}(\phi_I=1)\leq \alpha$, where $H_I=\bigcap_{i\in I} H_i$. Furthermore, the family of intersection tests $(\phi_I)_{I \subseteq \mathbb{N}}$ is consonant and predictable \citep{fischer2022online}. With this, Theorem 4.2 by Fischer et al. (2024) \citep{fischer2022online} implies that the corresponding closed procedure is defined by the individual significance levels $(\alpha_i^{I_i})_{i\in \mathbb{N}}$, where $I_1=\{1\}$ and $I_i=\{j\in \mathbb{N}: j<i, P_j>\alpha_j^{I_j}\} \cup \{i\}$ for all $i\geq 2$. This can also be written as
\begin{align*}
\alpha_i^{\text{c-graph}}=\alpha_i^{I_i}&=(\tau_i-\lambda_i)\left(\alpha \gamma_i + \sum_{j\in I_i, j<i-L_i} g_{j,i} (C_j-S_j+1) \frac{\alpha_j^{I_i}}{\tau_j-\lambda_j} +  \sum_{j\notin I_i, j<i} g_{j,i} \frac{\alpha_j^{I_i\cup \{j\}}}{\tau_j-\lambda_j}\right) \\
&= (\tau_i-\lambda_i)\left(\alpha \gamma_i + \sum_{j<i-L_i} g_{j,i} (1-R_j)(C_j-S_j+1) \frac{\alpha_j^{\text{c-graph}}}{\tau_j-\lambda_j} +  \sum_{ j<i} g_{j,i} R_j\frac{\alpha_j^{\text{c-graph}}}{\tau_j-\lambda_j}\right)\\
&=  (\tau_i-\lambda_i)\left(\alpha \gamma_i + \sum_{j=1}^{i-L_i-1} g_{j,i} (\max\{C_j,R_j\}-S_j+1) \frac{\alpha_j^{\text{c-graph}}}{\tau_j-\lambda_j} +  \sum_{j=i-L_i}^{i-1} g_{j,i} R_j\frac{\alpha_j^{\text{c-graph}}}{\tau_j-\lambda_j}\right),
\end{align*}
where $R_j=\mathbbm{1}\{P_j\leq \alpha_j^{\text{c-graph}}\}$.

\section{Exploiting the correlation structure when considering FWER control}\label{sec:correlation_structure}
\vspace{0.2cm}
When information about the joint distribution of $p$-values is available, ignoring this information might result in a conservative procedure \cite{westfall1993resampling}. Therefore, in this section we aim to incorporate such information into ADDIS procedures under local dependence. However, in order to do so we have to make some assumptions. First, we restrict to $\tau_i=1$ for all $i\in \mathbb{N}$. Therefore, we only adapt to the proportion of false null hypotheses and skip the discarding of conservative null $p$-values. 
% Setting $\lambda_i$ to a constant is a usual choice in most applications \cite{robertson2022online}. 
Furthermore, we assume that the hypotheses follow a batch dependence structure (see Section 5 for further explanation). We denote $b_i$, $i\in \mathbb{N}$, as the index of the batch that contains the $p$-value $P_i$. At last, we assume that the subset pivolatity condition holds \cite{westfall1993resampling}, which states that the distribution of $\mathbf{P_I}|H_I$ is the same as $\mathbf{P_I}|H_{\mathbb{N}}$ for every $I\subseteq \mathbb{N}$, where $\mathbf{P_I}=(P_i)_{i\in I}$ is a random vector of $p$-values. This is a very common assumption made when incorporating information about the joint distribution of $p$-values into multiple testing procedures and shown to hold in a wide range of settings, e.g. in a multivariate Gaussian \cite{westfall1993resampling}.

\begin{theorem}\label{theo:corr_princ}
    Assume that the subset pivotality condition is satisfied. Let the local dependence structure be given by batches $(B_i)_{i\in \mathbb{N}}$. Furthermore, let $\lambda_{b_i}\in [0,1)$ and $\alpha_i$ be measurable regarding $\mathcal{F}_{b_i-1}=\sigma(\{P_j\}_{j:b_j<b_i})$ for all $i\in \mathbb{N}$, where $R_j=\mathbbm{1}\{P_j\leq \alpha_j\}$ and $C_j=\mathbbm{1}\{P_j\leq \lambda_{b_j}\}$. Every multiple testing procedure controls the FWER in the strong sense when the individual significance levels $(\alpha_i)_{i\in \mathbb{N}}$ satisfy 
    \begin{align}
        \sum_{j=1}^{i} \frac{\alpha_j^{c}}{1-\lambda_{b_j}} (1-C_j) \leq \alpha \quad \text{ for all } i\in \mathbb{N},
    \end{align}
    where $\alpha_j^{c}\coloneqq\mathbb{P}_{H_\mathbb{N}}\left( \bigcap\limits_{k\in B_{b_j}, k<j, C_k=0} \{P_k>\alpha_k\} \cap \{P_j\leq \alpha_j\}  \Bigg \vert \mathcal{F}_{b_j-1}  \right)$ and $\mathbb{P}_{H_\mathbb{N}}$ indicates that we calculate the probability under the global null hypothesis.
\end{theorem}

\begin{remark}
  Note that this is a uniform improvement of the classical ADDIS principle under local dependence \cite{TR} for $\tau_i=1$ and batch-wise fixed $\lambda_i=\lambda_{b_i}$, since $\alpha_i^c\leq \alpha_i$ for all $i\in \mathbb{N}$.
\end{remark}

We propose the following Adaptive-Graph as example procedure that satisfies Theorem \ref{theo:corr_princ}.

\begin{definition}[$\text{Adaptive-Graph}_{\text{corr}}$]
Assume the local dependence structure is given by the batches $(B_i)_{i\in \mathbb{N}}$. Let $\lambda_{b_i} \in [0,1)$, $(\gamma_i)_{i\in \mathbb{N}}$ be a non-negative sequence that sums up to $1$ and $(g_{j,i}^*)_{i =j+1}^{\infty}$ be a non-negative sequence for all $j\in \mathbb{N}$ such that $g_{j,i}^*=0$ if $b_j=b_i$ and $\sum_{i:b_i>b_j}^{\infty} g_{j,i}^* \leq 1$. In addition, let $\lambda_{b_i}$, $\gamma_i$ and $g_{j,i}^*$, be measurable regarding $\mathcal{F}_{b_i-1}$. The  $\text{Adaptive-Graph}_{\text{corr}}$ tests each hypothesis $H_i$ at significance level 
\begin{align}\alpha_i=(1-\lambda_{b_i})\left(\alpha \gamma_i +  \sum_{j: b_j<b_i} g_{j,i}^* C_j \frac{\alpha_j}{1-\lambda_{b_j}} + \sum_{j: b_j<b_i} g_{j,i}^* (1-C_j) \frac{\alpha_j-\alpha_j^{c}}{1-\lambda_{b_j}} \right). \label{eq:graph_corr} \end{align}
\end{definition}

\begin{theorem}\label{theo:graph_corr}
    The $\text{Adaptive-Graph}_{\text{corr}}$ satisfies Theorem \ref{theo:corr_princ} and thus controls the FWER strongly when the subset pivotality condition is satisfied.
\end{theorem}

The $\text{Adaptive-Graph}_{\text{corr}}$ can be interpreted just as the $\text{ADDIS-Graph}_{\text{conf}}$ (Figure 4), however, if $P_i>\lambda_{b_i}$, the significance level $\alpha_i-\alpha_i^{c}$ is additionally distributed to the future hypotheses.

\begin{remark}
    In general, exploiting correlation structures in graphical test procedures is not straightforward, as the required consonance can get lost \cite{bretz2011graphical}. However, in the above described batch dependence setting, the here introduced $\text{Adaptive-Graph}_{\text{corr}}$ brings together the graphical approach and the utilization of information about the joint distribution of $p$-values.
\end{remark}

\begin{remark}
  If one chooses $B_1=\{P_1,P_2,\ldots \}$ and $\lambda_{b_1}=0$, the $\text{Adaptive-Graph}_{\text{corr}}$ no longer adapts to the number of false hypotheses, however, it allows to exploit the joint distribution among all hypotheses. This can be useful if there are no or only very few independent p-values.
\end{remark}

The batch setting assumed in Theorem \ref{theo:corr_princ} may seems very restrictive. However, it arises naturally in a lot of settings. For example, if the data for testing the hypotheses is replaced by new, independent data after a period of time. This is e.g. the case when a machine learning algortihm is updated over time and after testing several modifications, a new evaluation data set is used for future modifications \cite{FES, Fetal}.
Sometimes also platform trials are performed in a batch setting when multiple treatment arms enter and leave the trial at the same time \cite{Retal}. But also if this is not the case, platform trials can still be transformed into such a batch setting. 
% We propose two methods to do so. First, we note that the $p$-value used for $C_i$ does not have to be the $p$-value used for testing $H_i$. It is only required that $\mathbb{P}(P_i>\lambda|\mathcal{F}_{b_i-1})\geq (1-\lambda)$ if $i\in H_0$. Thus, we can also define $C_i\coloneqq\mathbbm{1}\{\tilde{P}_i\leq \lambda\}$, where $\tilde{P}_i$ is the $p$-value for $T_i$ based only on the data that is not used by future batches. For example, in Figure \ref{fig:platform_trial} one could batch the Treatment arms 1-4 to exploit the correlation between them and calculate $C_3$ and $C_4$ only based on data that does not overlap with $T_5$. 
For this, we specify a local dependence structure for batches and adjust the procedures in the same way as shown before for single hypotheses. For example, one could batch the $p$-values from Figure 2 as $B_1=\{P_1\}$, $B_2=\{P_2,P_3,P_4, P_5\}$ and $B_3=\{P_6,\ldots\}$. With this, the correlation within the batch $B_2$ could be exploited, but due to the dependence of $P_1$ and $P_2$ all the significance levels used for testing hypotheses in $B_2$ would not be allowed to use information about $P_1$. However, the significance levels for $B_3$ could depend on $P_1$. This could save a lot of significance level, particularly, if the testing process continues after $T_6$.

\section{Extension to FDR control}\label{sec:FDR}
\vspace{0.2cm}

Tian \& Ramdas (2019) \citep{TR2} introduced the following ADDIS condition for FDR control 
\begin{align}
\frac{{}\sum_{j=1}^{i} \frac{\alpha_j}{\tau_j-\lambda_j} (S_j-C_j)}{|R(i)|\lor 1}
\leq \alpha
 \quad \text{for all } i\in \mathbb{N}.
\label{eq:FDR_ADDIS_principle}
\end{align}

The only difference to the ADDIS condition for the FWER control  (1) is the denominator $|R(i)|\lor 1$. Bringing it on the other side, it can be interpreted as if  an additional level $\alpha$ is gained after each rejection except for the first one. This can be incorporated into the ADDIS-Graph by distributing an additional $\alpha$ to future hypotheses in case of rejection according to non-negative weights $(h_{j,i})_{i=j+1}^{\infty}$ such that $\sum_{i=j+1}^{\infty} h_{j,i}\leq 1$ for all $j\in \mathbb{N}$. For example, one could just choose $h_{j,i}=g_{j,i}$.

% In order to control FDR($i$) at any time $i\in \mathbb{N}$ using ADDIS procedures, we need the additional assumptions that $\lambda_i \geq \alpha_i$ for all $i\in \mathbb{N}$ and that $\alpha_i$, $\lambda_i$ and $1-\tau_i$ are monotonic functions of the past. This means that they are coordinatewise nondecreasing functions in $R_{1:(i-1)}\coloneqq (R_1,\ldots,R_{i-1})$ and $C_{1:(i-1)}\coloneqq (C_1,\ldots,C_{i-1})$ and nonincreasing in $S_{1:(i-1)}\coloneqq (S_1,\ldots,S_{i-1})$. Under these assumptions, Tian \& Ramdas\citep{TR2} showed that the FDR is controlled if the ADDIS condition \eqref{eq:cond_addis_principle} for FWER control is replaced with 
% \begin{align}
% \frac{{}\sum_{j=1}^{i} \frac{\alpha_j}{\tau_j-\lambda_j} (S_j-C_j)}{|R(i)|\lor 1}
% \leq \alpha
%  \quad \text{for all } i\in \mathbb{N}.
% \label{eq:FDR_ADDIS_principle}
% \end{align}

% \begin{remark}
%  If $\tau_i=1$ for all $i\in \mathbb{N}$, the ADDIS condition for FDR control reduces to the SAFFRON condition\citep{RZWJ}. In this case, the uniformly validity assumption of the null $p$-values can be dropped.
% \end{remark} 

Since no significance level is gained for the first rejection, FDR procedures often assume that a lower overall significance level of $W_0\leq \alpha$ is available at the beginning of the testing process such that $(\alpha-W_0)$ can be gained after the first rejection. To differentiate between the first and other rejections, we additionally define the indicator $K_i$ with $K_i=1$, if the first rejection happened within the first $i-1$ steps and $K_i=0$, otherwise. We also set $K_i^c=1-K_i$. With this, the ADDIS-Graph for FDR control can be defined as follows. 
\begin{definition}[$\text{FDR-ADDIS-Graph}_{\text{conf}}$\label{def:addis_graph_fdr}]
Let the conflict sets be given by $(\mathcal{X}_i)_{i\in \mathbb{N}}$. Furthermore, let $(\gamma_i)_{i\in \mathbb{N}}$, $(g_{j,i}^*)_{j\in \mathbb{N}, i>j}$, $(\tau_i)_{i\in \mathbb{N}}$ and $(\lambda_i)_{i\in \mathbb{N}}$ be as in $\text{ADDIS-Graph}_{\text{conf}}$ (Definition 2). In addition, let $W_0\leq \alpha$ and $(h_{j,i}^*)_{i =j+1}^{\infty}$, $j\in \mathbb{N}$, be a non-negative sequence such that $h_{j,i}^*=0$ if $j\in \mathcal{X}_i$ and $\sum_{i>j,j\notin \mathcal{X}_i}^{\infty} h_{j,i}^*\leq 1 $. The $\text{FDR-ADDIS-Graph}_{\text{conf}}$ tests each hypothesis $H_i$ at significance level $\alpha_i=\min(\hat{\alpha}_i,\lambda_i)$, where
\begin{align}
\hat{\alpha}_i = (\tau_i-\lambda_i)\left(W_0 \gamma_i + \sum_{j=1}^{i-1} g_{j,i}^*(C_j-S_j+1)  \frac{\hat{\alpha_j}}{\tau_i-\lambda_j} + \sum_{j=1}^{i-1} h_{j,i}^* R_j [\alpha K_j + (\alpha-W_0)K_j^c] \right) 
\label{eq:fdr_addis_graph}\end{align}
with $R_j=\mathbbm{1}\{P_j\leq \alpha_j\}$.
\end{definition}

% There are several results for FDR and mFDR control when condition \eqref{eq:FDR_ADDIS_principle} is fulfilled that are based on different assumptions about the parameter $\tau_i$ and the dependence structure of the $p$-values\citep{RZWJ,TR2,ZRJ,fisher2021online}. In particular, the two following conditions are important when considering FDR control.

In order to control the FDR using ADDIS procedures, $\alpha_i$, $\lambda_i$ and $1-\tau_i$, $i\in \mathbb{N}$, are required to be monotonic functions of the past \citep{TR2}. This means that they are coordinatewise nondecreasing functions in $R_{1:(i-1)}\coloneqq (R_1,\ldots,R_{i-1})$ and $C_{1:(i-1)}\coloneqq (C_1,\ldots,C_{i-1})$ and nonincreasing in $S_{1:(i-1)}\coloneqq (S_1,\ldots,S_{i-1})$.
% \begin{condition}[Monotonicity\citep{RZWJ,TR2}]
%     The parameters $\alpha_i$, $\lambda_i$ and $1-\tau_i$ are coordinatewise nondecreasing functions in $R_{1:(i-1)}\coloneqq (R_1,\ldots,R_{i-1})$ and $C_{1:(i-1)}\coloneqq (C_1,\ldots,C_{i-1})$ and nonincreasing in $S_{1:(i-1)}\coloneqq (S_1,\ldots,S_{i-1})$.
% \end{condition}
An easy way to satisfy this using the $\text{FDR-ADDIS-Graph}_{\text{conf}}$ is to choose the parameters $\lambda_i$, $\tau_i$, $\gamma_i$, $g_{j,i}^*$ and $h_{j,i}^*$ for all $i\in \mathbb{N}, j<i$, independently of the past. Then $\alpha_i$ is a monotonic function of the past by definition.

\begin{theorem}\label{theo:addis_graph_fdr}
The $\text{FDR-ADDIS-Graph}_{\text{conf}}$ satisfies equation \eqref{eq:FDR_ADDIS_principle}. Thus, it controls the mFDR($i$) for all $i\in \mathbb{N}$ \citep{ZRJ}.  Furthermore, it controls the FDR($i$) for all $i\in \mathbb{N}$ when $\alpha_i$, $\lambda_i$ and $1-\tau_i$ are monotonic functions of the past and the null p-values are independent from each other and the non-nulls \citep{TR2}.

% \begin{enumerate}[label=(\alph*)]
%     \item $\tau_i=1$ for all $i\in \mathbb{N}$\citep{RZWJ, ZRJ} or \label{a}
%     \item the null $p$-values are uniformly valid\citep{TR2,ZRJ}.\label{b}
% \end{enumerate}
%    Furthermore, it controls the FDR($i$) for all $i\in \mathbb{N}$ when $\alpha_i$, $\lambda_i$ and $1-\tau_i$ are monotonic functions of the past and 
% \begin{enumerate}[resume, label=(\alph*)]
%     % \item $\tau_i=1$ for all $i\in \mathbb{N}$ and the $p$-values are conditional PRDS\citep{fisher2021online} or \label{c}
%     \item the null $p$-values are uniformly valid and independent from each other and the non-nulls \citep{TR2}. \label{d}
% \end{enumerate}
\end{theorem}

% Thus, if we aim for FDR control in a platform trial, we can use an $\text{FDR-ADDIS-Graph}_{\text{conf}}$, where $\alpha_i$ and $\lambda_i$ are monotonic functions of the past and $\tau_i=1$ for all $i\in \mathbb{N}$. In contrast, mFDR control would only require that the null $p$-values are uniformly valid (or $\tau_i=1$).

The $\text{FDR-ADDIS-Graph}_{\text{conf}}$ is illustrated in Figure \ref{abb:addis_graph_fdr}. Note that the figure only contains $(\hat{\alpha}_i)_{i\in \mathbb{N}}$  and one needs to set $\alpha_i=\min(\hat{\alpha}_i,\lambda_i)$ after using the graph. The $\text{FDR-ADDIS-Graph}_{\text{conf}}$ can be interpreted just as the $\text{ADDIS-Graph}_{\text{conf}}$ for FWER control (Figure 4). The additional grey arrows are activated if the corresponding hypothesis is rejected. In case of the first rejection, the level $\alpha-W_0$ is distributed to the future hypotheses according to the weights $(h_{j,i})_{i=j+1}^{\infty}$, $j\in \mathbb{N}$, and in case of any other rejection, the level $\alpha$ is distributed. 

\begin{figure}[htbp]
	\begin{center}
			\centering
		\includegraphics[width=21cm,height=7cm,keepaspectratio]{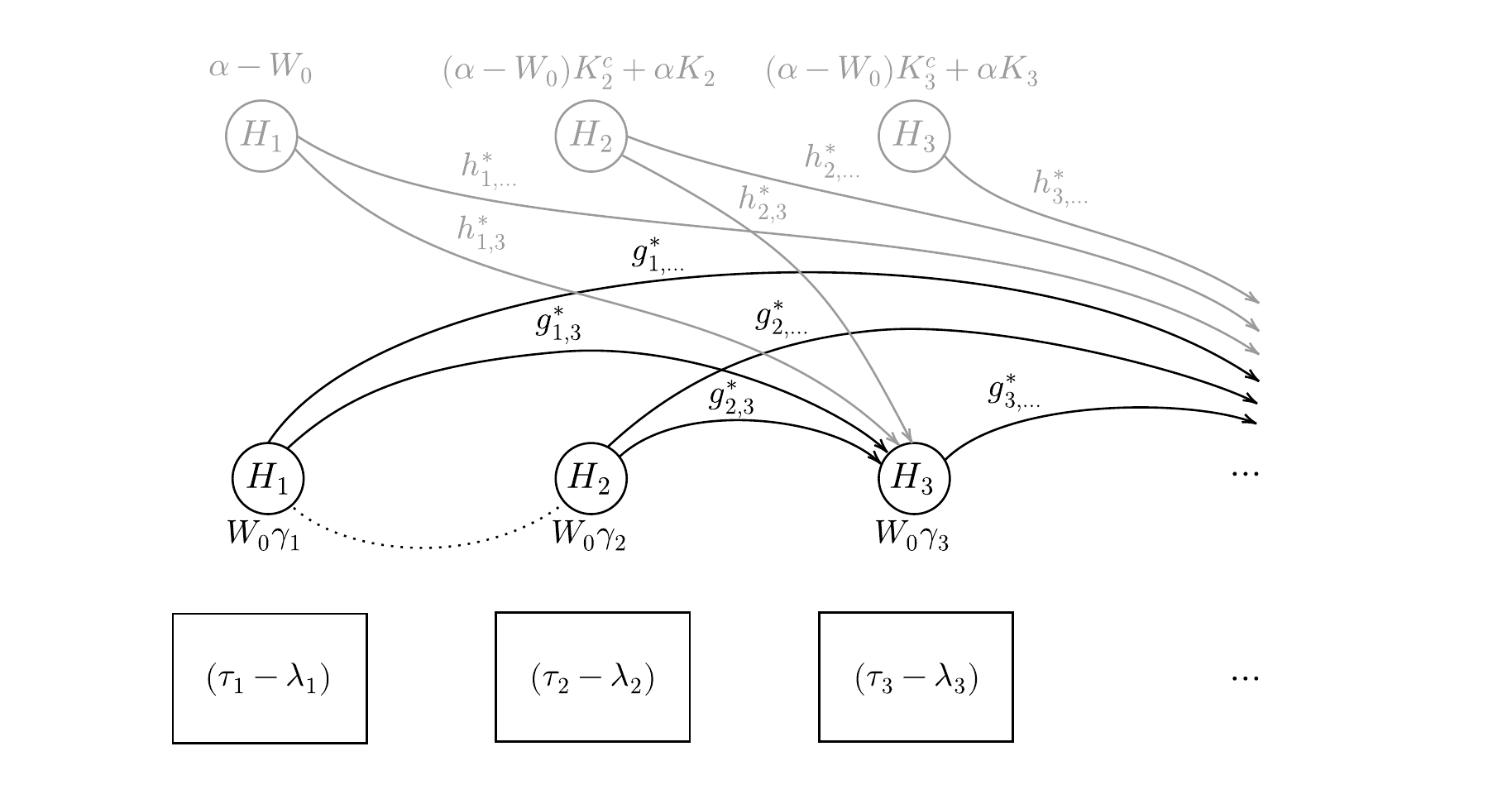}
	\end{center}
	\caption{Illustration of the FDR-ADDIS-Graph.\label{abb:addis_graph_fdr}}
\end{figure}

The benefit of the $\text{FDR-ADDIS-Graph}_{\text{conf}}$ compared to the current state-of-art ADDIS procedure for FDR control, the $\text{ADDIS}^*$ algorithm \cite{TR2}, is similar as for the $\text{FDR-ADDIS-Graph}_{\text{conf}}$ for FWER control and the $\text{ADDIS-Spending}_{\text{local}}$. Due to its graphical structure, the $\text{FDR-ADDIS-Graph}_{\text{conf}}$ is more flexible and easier to interpret. In particular, the dependencies between the previous test outcomes and individual significance levels become clearer. This might be even more important in the FDR case, as the $\text{ADDIS}^*$ is far more complex than the $\text{ADDIS-Spending}_{\text{local}}$. Although we do not show a uniform improvement theoretically, the intuition about the superiority of the $\text{FDR-ADDIS-Graph}_{\text{conf}}$ over the $\text{ADDIS}^*$ algorithm when conflicts are present is the same as in the FWER case: The $\text{FDR-ADDIS-Graph}_{\text{conf}}$ distributes the same amount of significance level under conflict sets, while the $\text{ADDIS}^*$ algorithm loses level systematically due to conflicts. In the section \ref{sec:sim_FDR}, we verify this by means of simulations and quantify the resulting power difference. Furthermore, these graphical representations clarify the gain of switching from FWER to FDR control.

\section{Further simulation results}\label{sec:sim_appendix}
\vspace{0.2cm}

\subsection{Comparison of the ADDIS-Graph and closed ADDIS-Spending under local dependence}\label{sec:sim_closed_ADDIS}

In this subsection, we compare the closed $\text{ADDIS-Spending}_{\text{local}}$ (7) with the  $\text{ADDIS-Graph}_{\text{conf-u}}$ (3). We use exactly the same simulation setup as in Section 7 and apply the procedures with the same parameters. 

The results are illustrated in Figure \ref{fig:plot_fwer_closed} and look very similar as in Figure 5. The only difference is that in Figure \ref{fig:plot_fwer_closed} the closed $\text{ADDIS-Spending}_{\text{local}}$ loses slightly less power due to the local dependence than the $\text{ADDIS-Spending}_{\text{local}}$. However, the $\text{ADDIS-Graph}_{\text{conf-u}}$ also outperforms the closed $\text{ADDIS-Spending}_{\text{local}}$ in all cases.

\begin{figure}[htbp]
	\begin{center}
			\centering
		\includegraphics[width=19.5cm,height=6.5cm,keepaspectratio]{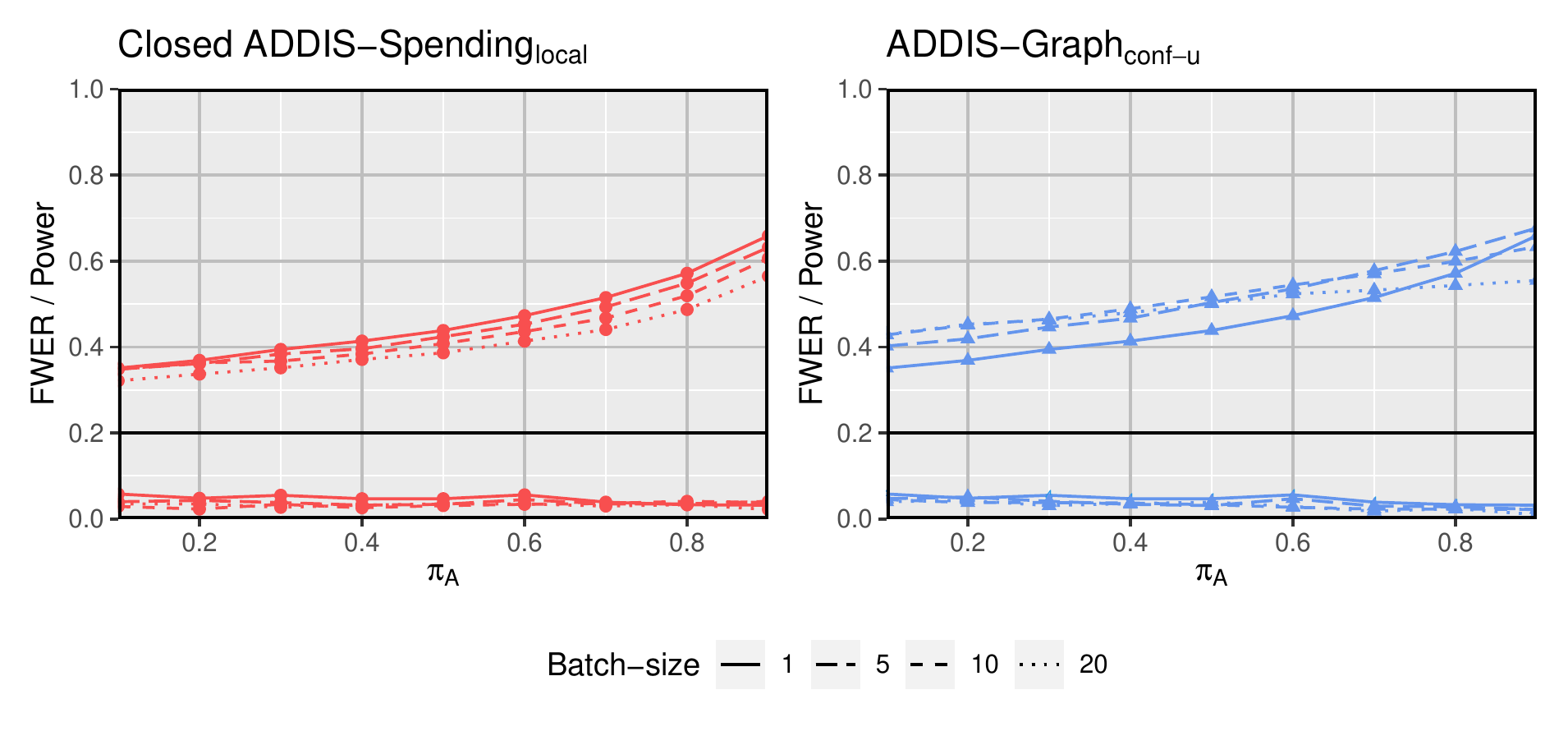}
  \includegraphics[width=19.5cm,height=6.5cm,keepaspectratio]{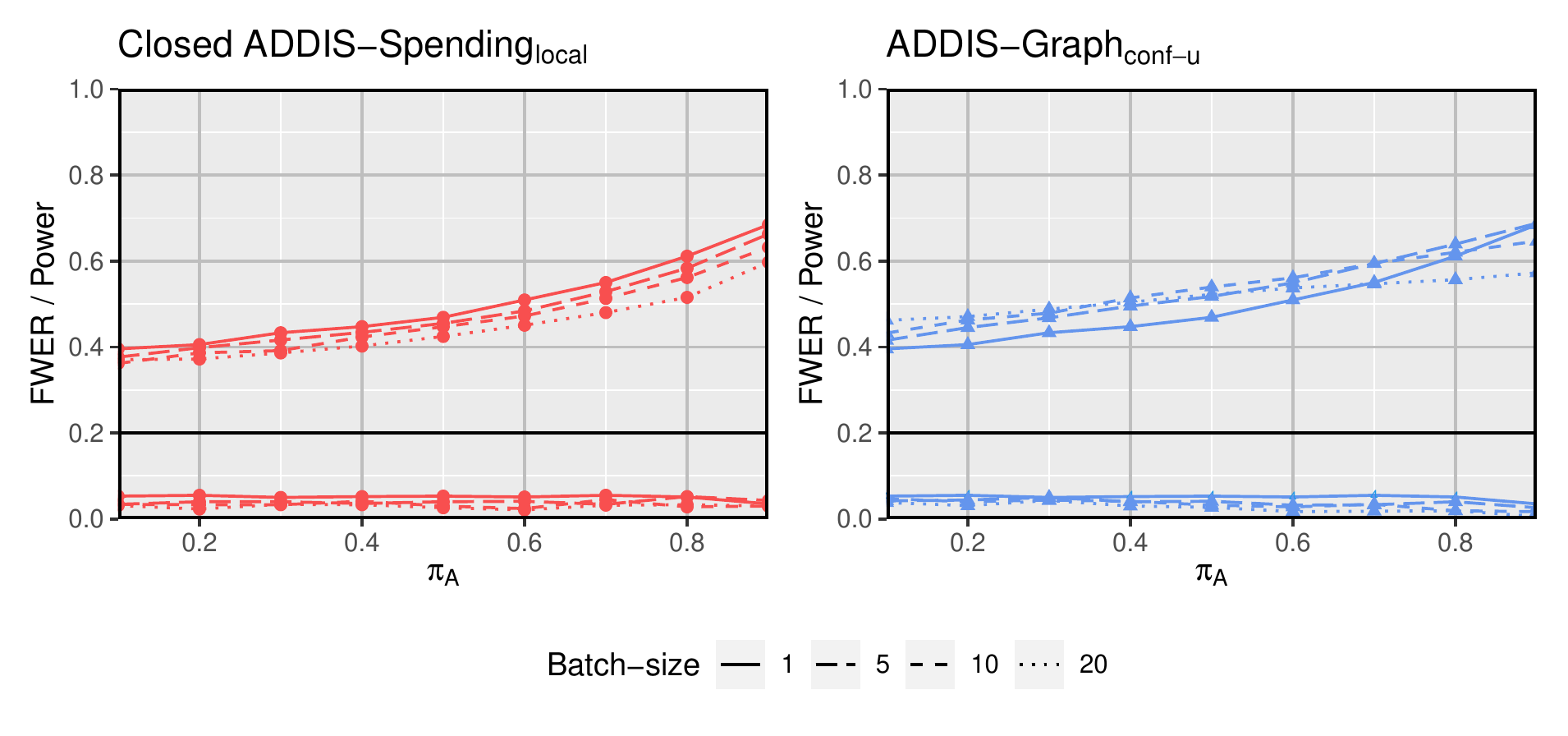}
  \includegraphics[width=19.5cm,height=6.5cm,keepaspectratio]{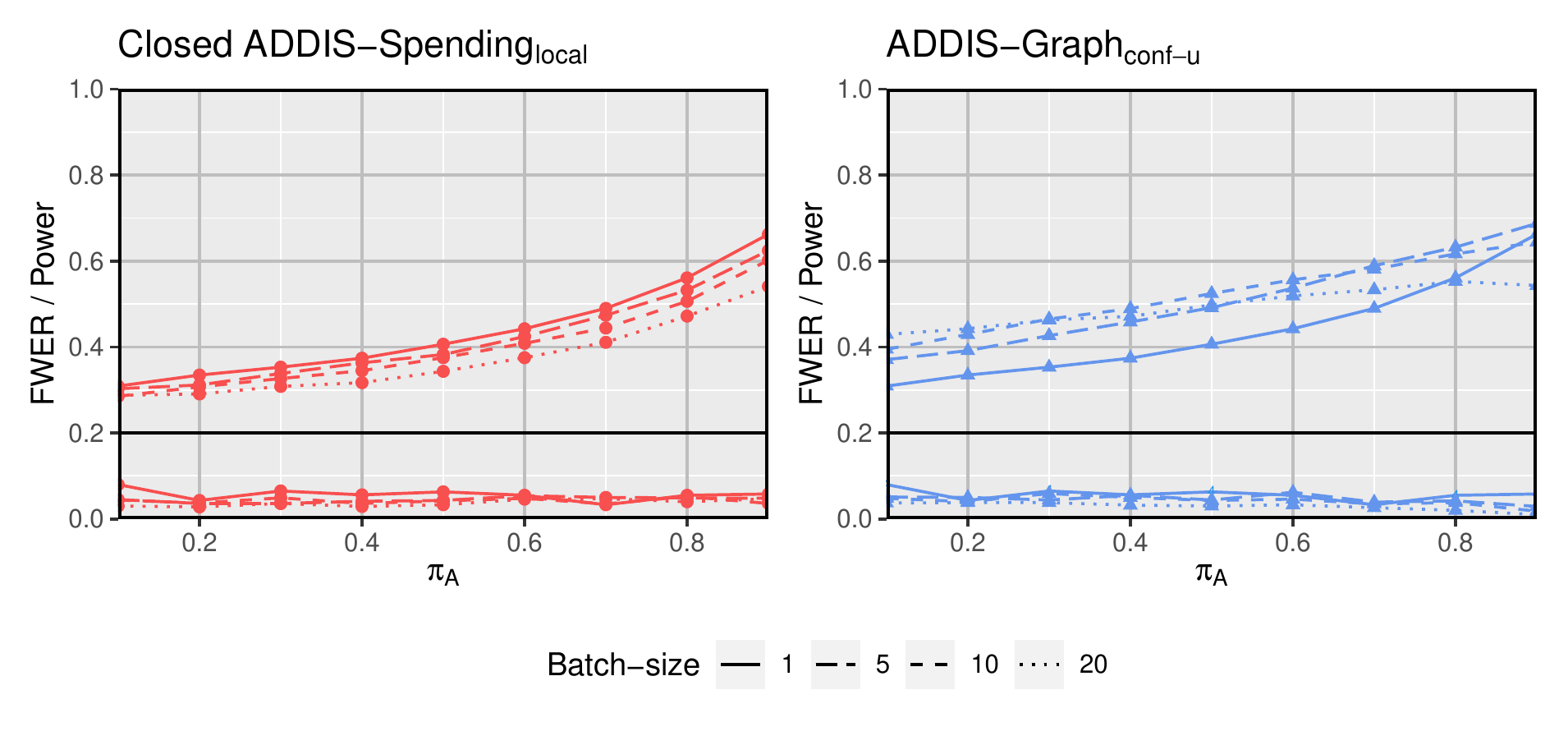}
	\end{center}
	\caption{Comparison of closed $\text{ADDIS-Spending}_{\text{local}}$ (7) and $\text{ADDIS-Graph}_{\text{conf-u}}$ (3) in terms of power and FWER for different batch-sizes and proportions of false null hypotheses ($\pi_A$). Lines above the overall level $\alpha=0.2$ correspond to power and lines below to FWER. The $p$-values were generated as described in the text with parameters $\mu_N=-0.5$ and $\rho=0.5$. Both procedures were applied with parameters $\tau_i=0.8$ and $\lambda_i=0.16$. In the top row $\gamma_i\propto 1/\left((i+1)\log(i+1)^2\right)$, in the middle row $\gamma_i\propto 1/i^{1.6}$ and in the bottom row $\gamma_i=6/(\pi^2 i^2)$. Under independence of the $p$-values both procedures coincide. However, the closed $\text{ADDIS-Spending}_{\text{local}}$ loses power when the $p$-values become locally dependent, while the $\text{ADDIS-Graph}_{\text{conf-u}}$ offers a similar or even larger power as under independence. \label{fig:plot_fwer_closed}}
\end{figure}

\subsection{Simulation results when incorporating correlation structure}\label{sec:sim_corr}
In this subsection, we use the same simulation setup as described in Section 7 to compare the $\text{ADDIS-Graph}_{\text{conf}}$ (3) with the $\text{Adaptive-Graph}_{\text{corr}}$ \eqref{eq:graph_corr}. The results are summarized in Figure \ref{fig:plot_corr_batch}. In the top row, we vary the batch-size $b\in \{1,5,10,20\}$, in the middle row  the correlation within batches $\rho\in \{0.3,0.5,0.7,0.9\}$ and in the bottom row, we evaluate the procedures for a different conservativeness of null $p$-values $\mu_N \in \{0,-0.5,-1,-2\}$, while the other parameters are set to standard values $b=10$, $\rho=0.5$ and $\mu_N=0$. In order to extract the effect of incorporating the correlation structure, we set $\tau_i=1$ for the $\text{ADDIS-Graph}_{\text{conf}}$ in the top and middle row, which is why we also write $\text{Adaptive-Graph}_{\text{conf}}$.
Furthermore, we choose $\lambda_i=\tau_i \alpha$, $\gamma_i=6/(\pi^2 i^2)$ and $g_{j,i}^* = g_{j,i}\bigg/ \left(1-\sum_{k=j+1}^{d_j-1} g_{j,k}\right)$ if $i\geq d_j$ and $g_{j,i}^{*}=0$ otherwise, where $g_{j,i}=\gamma_{i-j}$ and $d_j=\min\{i\in \mathbb{N}:i-L_i>j\}$, for both procedures in all cases. 

\begin{figure}[htbp]
	\begin{center}
			\centering
		\includegraphics[width=19.5cm,height=6.5cm,keepaspectratio]{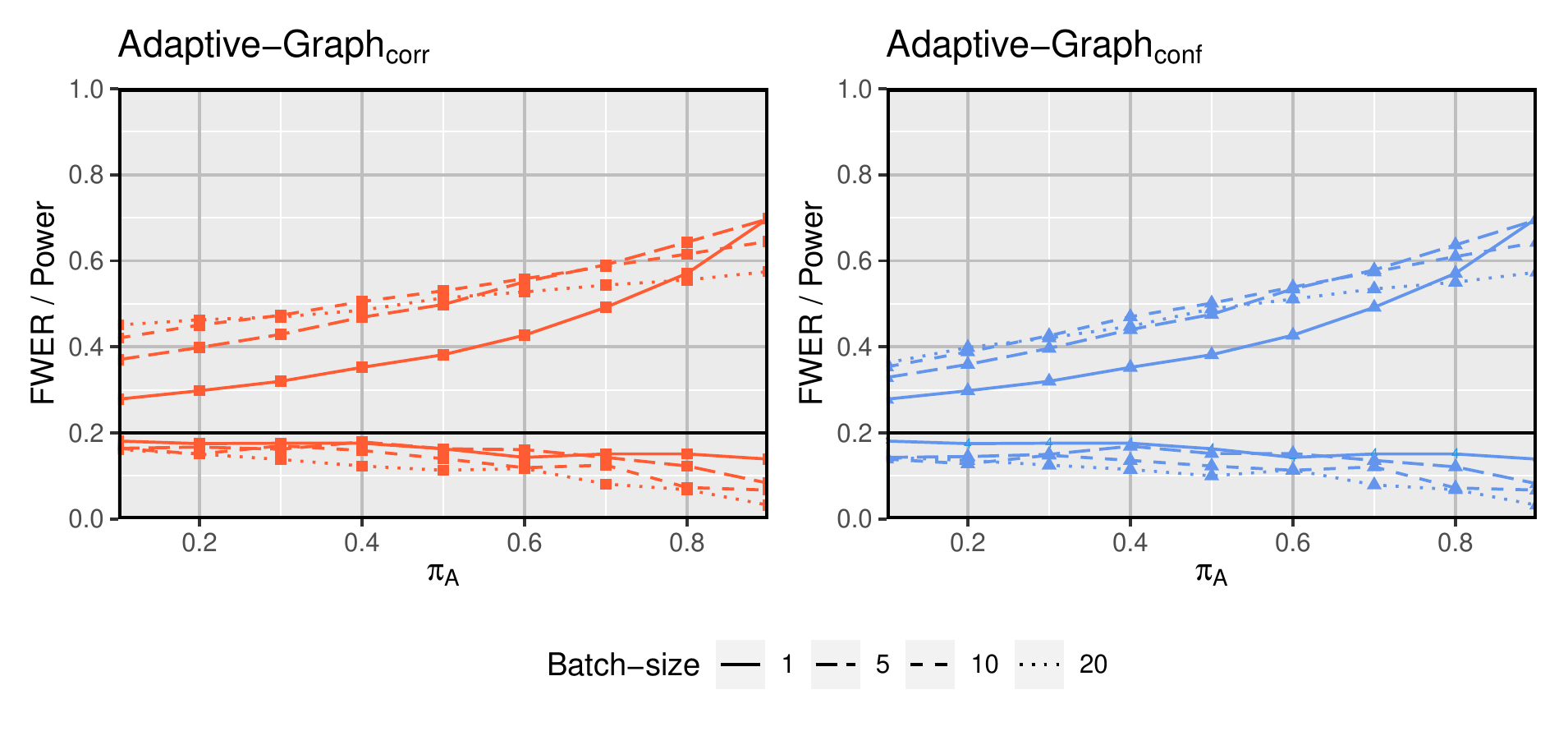}
\includegraphics[width=19.5cm,height=6.5cm,keepaspectratio]{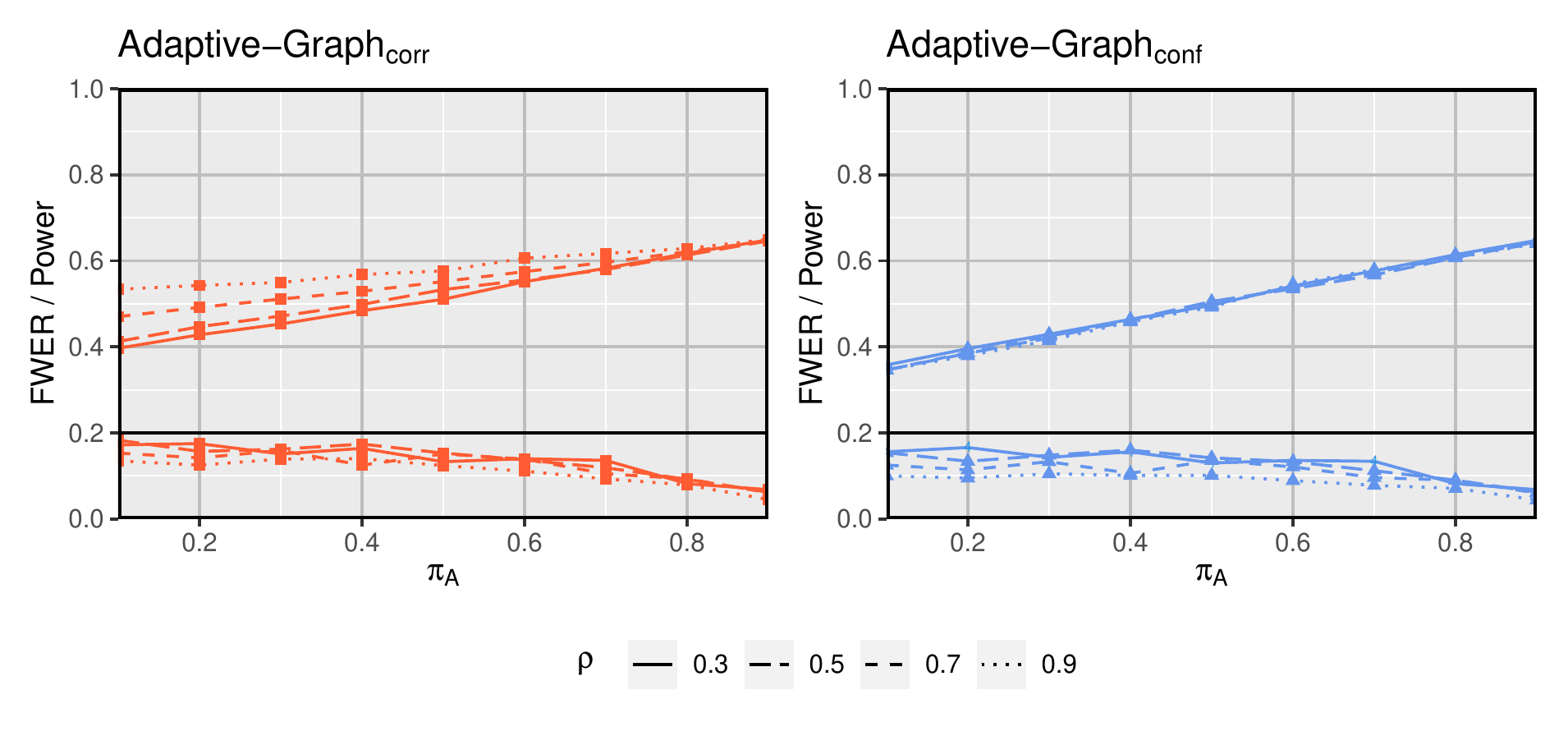}
\includegraphics[width=19.5cm,height=6.5cm,keepaspectratio]{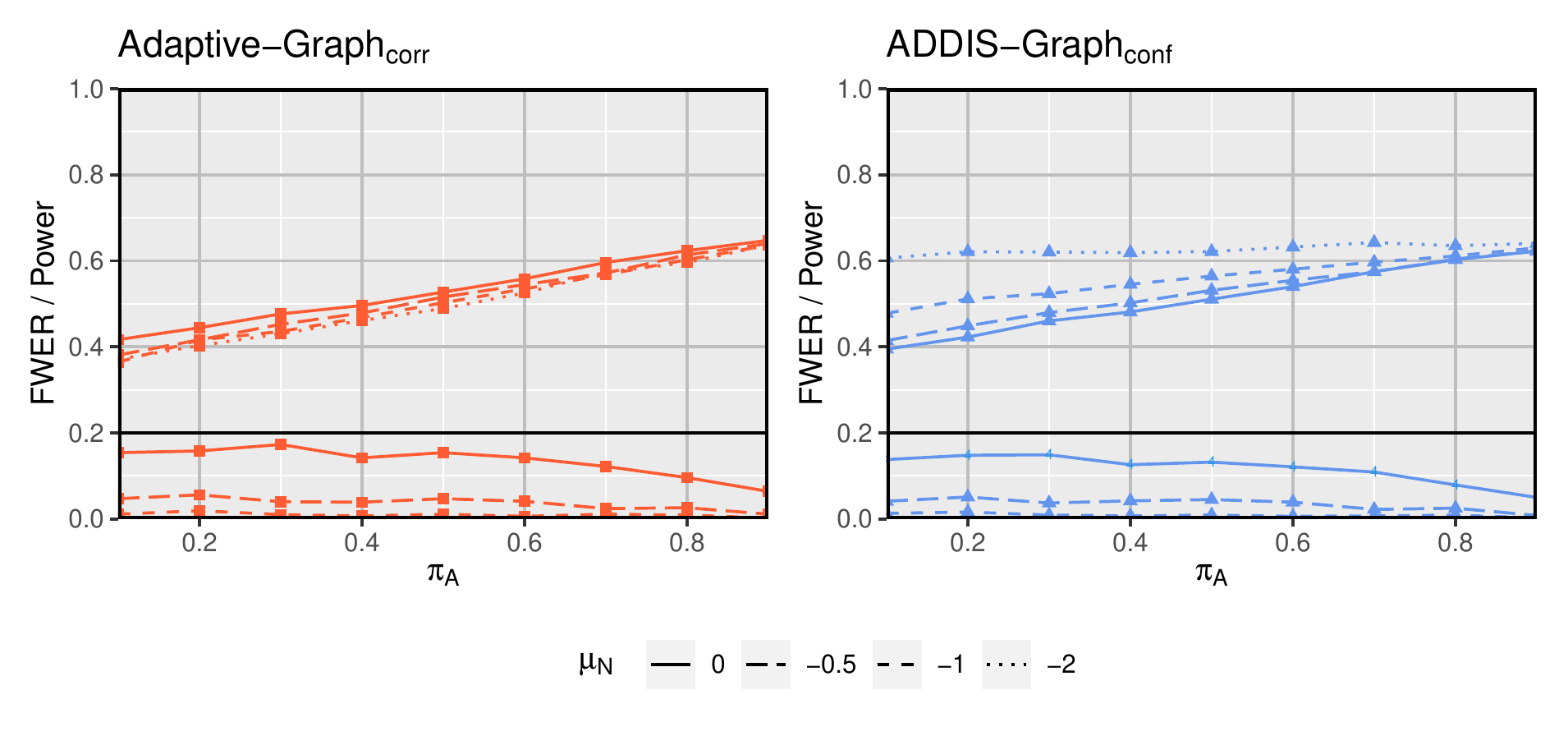}
	\end{center}
	\caption{Comparison of $\text{Adaptive-Graph}_{\text{corr}}$ \eqref{eq:graph_corr}  and $\text{Adaptive-Graph}_{\text{conf}}$ (3) in terms of power and FWER for different proportions of false null hypotheses ($\pi_A$). In the top row, we vary the batch-size; In the middle row, we vary the strength of correlation $\rho$; In the bottom row, we vary the conservativeness of the null p-values $\mu_N$. Lines above the overall level $\alpha=0.2$ correspond to power and lines below to FWER. The $p$-values were generated and the procedures applied as described in the text. The $\text{Adaptive-Graph}_{\text{corr}}$ allows to increase the power when the batch-size is large and the within-batch correlation is high. However, when the p-values become conservative, the $\text{Adaptive-Graph}_{\text{conf}}$ allows to gain power, while the $\text{Adaptive-Graph}_{\text{corr}}$ loses a bit of power.
 \label{fig:plot_corr_batch}}
\end{figure}

In the top row, the two procedures are equivalent under independence. However, when the batch-size increases, the power of both procedures increases as well, while the power gain is larger using the $\text{Adaptive-Graph}_{\text{corr}}$. The plots looks different in the middle row, where the power of the $\text{Adaptive-Graph}_{\text{conf}}$ remains identical when varying the $\rho$, while the FWER drops a bit for large $\rho$. This FWER drop can be compensated by exploiting the correlation structure using $\text{Adaptive-Graph}_{\text{corr}}$. It also seems that the strength of correlation in the middle row has a larger positive impact on the power of the $\text{Adaptive-Graph}_{\text{corr}}$ than the batch-size in the top row. The comparison looks quite different in the bottom row, where conservative $p$-values are now discarded using the $\text{ADDIS-Graph}_{\text{conf}}$. When the null $p$-values are uniformly distributed ($\mu_N=0$), the $\text{Adaptive-Graph}_{\text{corr}}$ is still superior, however, when the null $p$-values become conservative, the power of the $\text{ADDIS-Graph}_{\text{conf}}$ increases, while the $\text{Adaptive-Graph}_{\text{corr}}$ loses a bit of power. To conclude, it is difficult to give a general advice on whether one should prefer the  $\text{Adaptive-Graph}_{\text{corr}}$ or $\text{ADDIS-Graph}_{\text{conf}}$ and the choice of the procedure should incorporate information/assumptions about the batch-size, strength of correlation and conservativeness of null $p$-values.

\subsection{Comparison of FDR-ADDIS-Graph and $\text{ADDIS}^*$ in an asynchrone test setup}\label{sec:sim_FDR}

In this subsection we consider a similar simulation setup as described in Section 7, but for independent $p$-values ($b=1$) and a larger number of hypotheses ($n=1000$). Applying the procedures, it is assumed that the hypotheses are tested in an asynchronous manner. Thus, the conflict sets are given by $\mathcal{X}_i=\{j<i:E_j\geq i\}$,  where $E_i\geq i$ is the (possibly random but independent of $P_i$) stopping time for hypothesis $H_i$. Due to  Theorem \ref{theo:addis_graph_fdr}, the $\text{FDR-ADDIS-Graph}_{\text{conf}}$ controls the FDR in this setting. We assume that $E_i=i+e$ for some constant test duration $e\in \mathbb{N}_0$. In the following simulations we compare the $\text{FDR-ADDIS-Graph}_{\text{conf}}$ and $\text{ADDIS}^*_{\text{async}}$ \citep{TR2} in terms of power and FDR for $e\in \{0,2,5,10\}$. Since FDR is less conservative than FWER, we also change the overall level to $\alpha=0.05$. As recommended \citep{TR2}, we choose $\tau_i=0.5$ and $\lambda_i=0.25$ for all $i\in \mathbb{N}$, but use the same $(\gamma_i)_{i\in \mathbb{N}}$ as in Section 7. For the weights of the $\text{FDR-ADDIS-Graph}_{\text{conf}}$, we choose
$g_{j,i}^* = g_{j,i}\bigg/ \left(1-\sum_{k=j+1}^{E_j} g_{j,k}\right)$ if $i> E_j$ and $g_{j,i}^{*}=0$ otherwise, where $g_{j,i}=\gamma_{i-j}$ and $h_{j,i}^*=g_{j,i}^*$.
Furthermore, we set $W_0=\alpha$. The results obtained by averaging over $200$ independent trials can be found in the Figure \ref{fig:plot_fdr_logq}.

\begin{figure}[htbp]
	\begin{center}
			\centering
		\includegraphics[width=19.5cm,height=6.5cm,keepaspectratio]{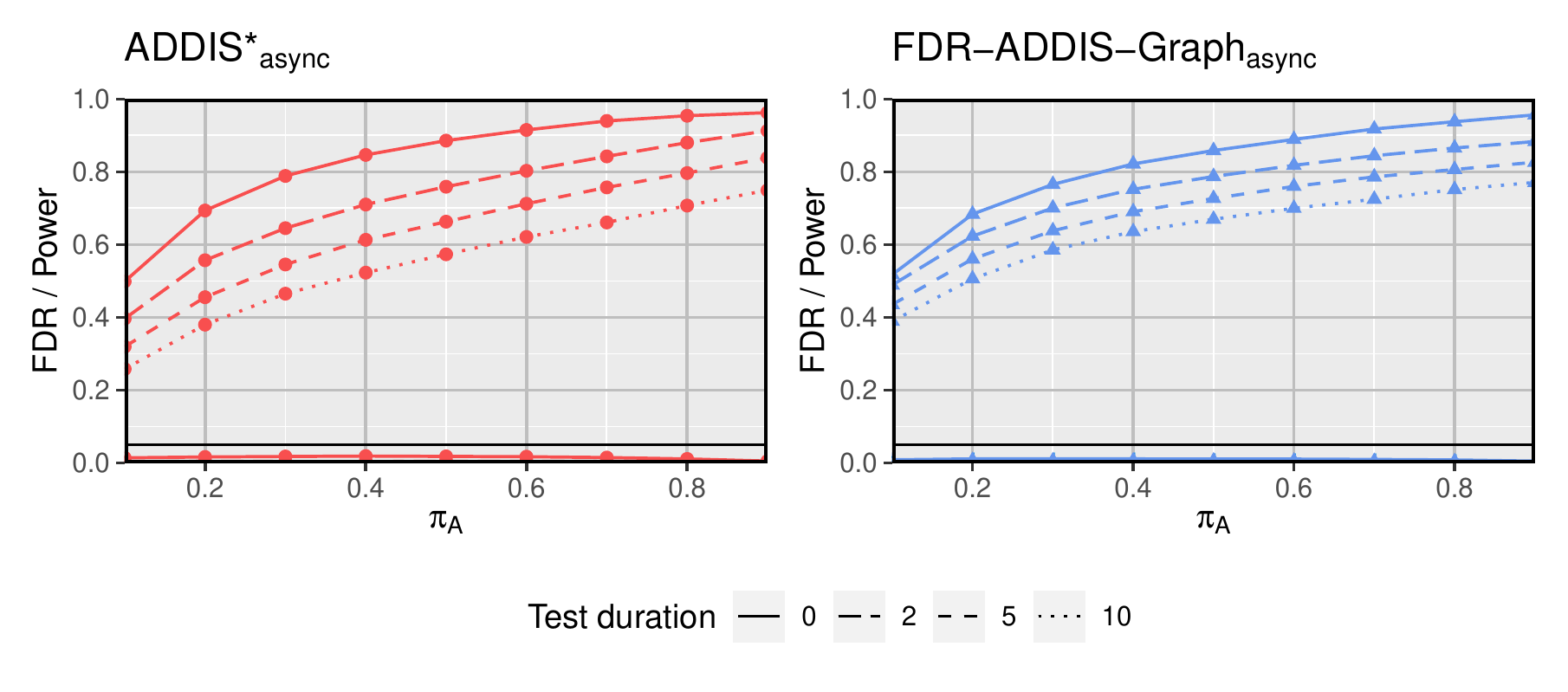}
  \includegraphics[width=19.5cm,height=6.5cm,keepaspectratio]{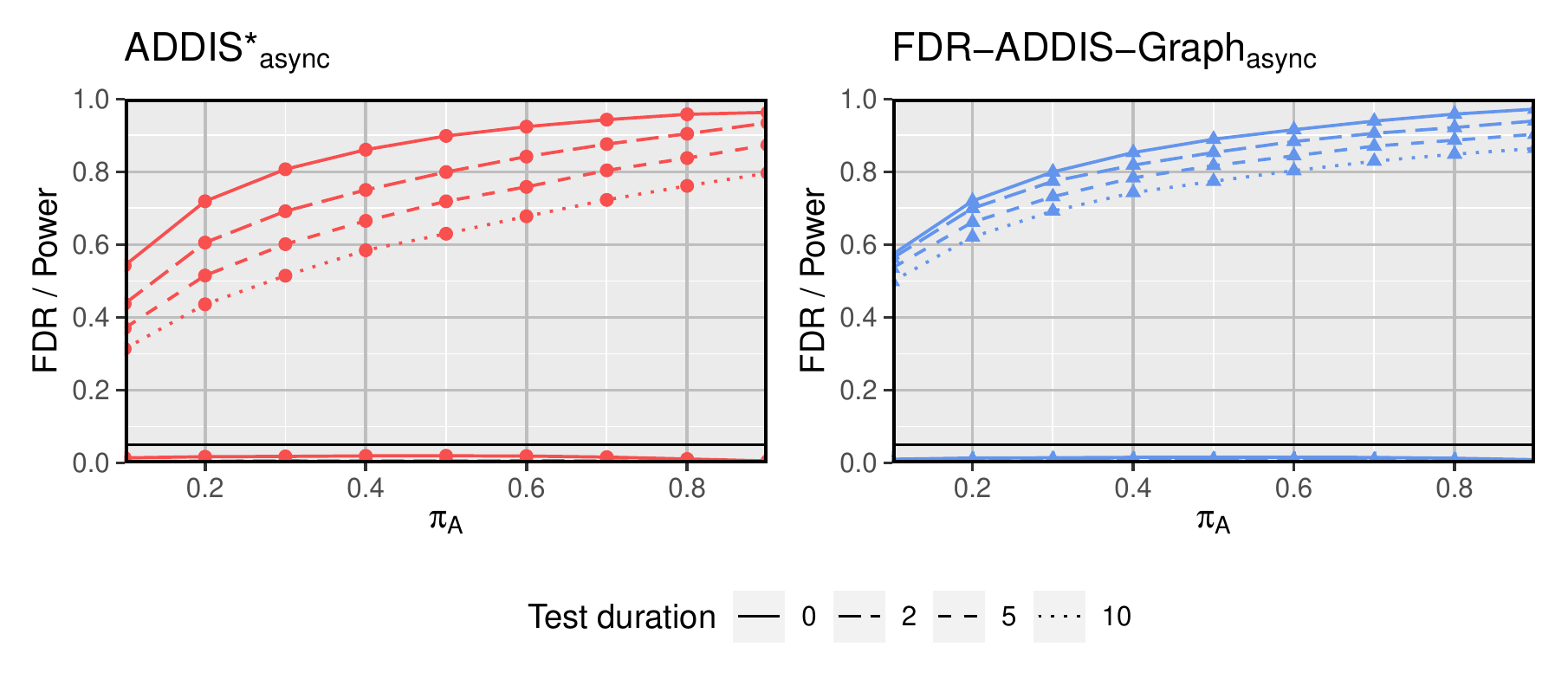}
  \includegraphics[width=19.5cm,height=6.5cm,keepaspectratio]{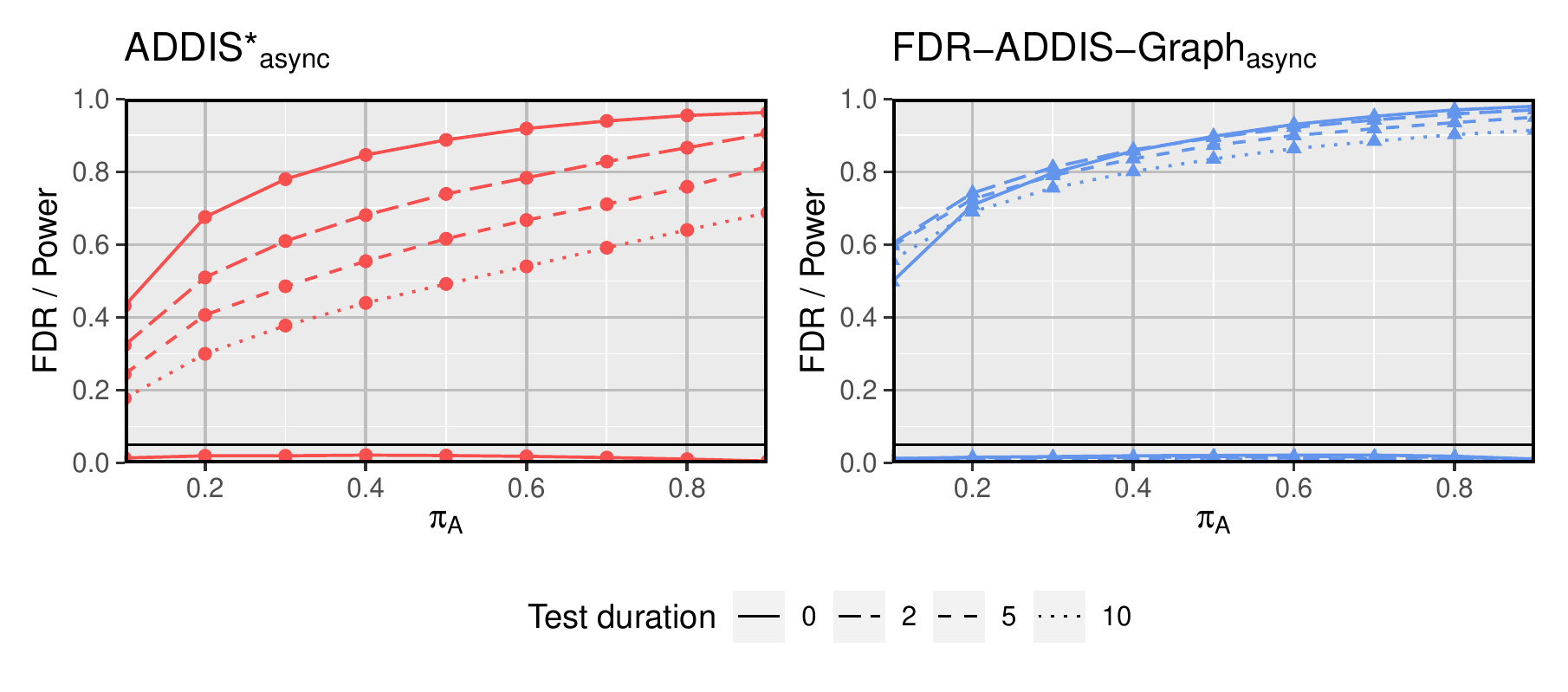}
	\end{center}
	\caption{Comparison of $\text{ADDIS}^*_{\text{async}}$ and $\text{FDR-ADDIS-Graph}_{\text{conf}}$ in terms of power and FDR for different test durations and proportions of false null hypotheses ($\pi_A$). Lines above the overall level $\alpha=0.05$ correspond to power and lines below to FDR. The $p$-values were generated as described in the text with parameter $\mu_N=-0.5$. Both procedures were applied with parameters $\tau_i=0.5$, $\lambda_i=0.25$ and $W_0=\alpha$. In the top row $\gamma_i\propto 1/\left((i+1)\log(i+1)^2\right)$, in the middle row $\gamma_i\propto 1/i^{1.6}$ and in the bottom row $\gamma_i=6/(\pi^2 i^2)$. When hypotheses are not tested asynchronously, the power of both procedures is similar. However, the $\text{ADDIS-Spending}_{\text{local}}$ loses power when the test duration increases, while the $\text{ADDIS-Graph}_{\text{conf-u}}$ can decelerate this decrease.
 \label{fig:plot_fdr_logq}}
\end{figure}

The results are similar as for the FWER controlling procedures (Section 7). The power of $\text{ADDIS}^*_{\text{async}}$ decreases enormously for an increasing test duration. This decrease can be decelerated by the $\text{FDR-ADDIS-Graph}_{\text{async}}$ (top and middle row) or even stopped (bottom row), if a faster decreasing $(\gamma_i)_{i\in \mathbb{N}}$ is chosen.

\section{Proofs}\label{sec:proofs}
\vspace{0.2cm}
\begin{proof}[Proof of Theorem 1]
Let $(\alpha_i)_{i\in \mathbb{N}}$ be given by the ADDIS-Graph. We need to show that for any $i\in \mathbb{N}$, $S_{1:i}\coloneqq (S_1,\ldots,S_i)^T \in \{0,1\}^{i-1}$ and $C_{1:i}\coloneqq (C_1,\ldots,C_i)^T \in \{0,1\}^{i}$:  \begin{align} \sum_{j=1}^{i} \frac{\alpha_i}{\tau_i-\lambda_i}(S_i-C_i)\leq \alpha.  \label{adavora}\end{align}  
 We define $U_j\coloneqq C_j-S_j+1$ for all $j\in \mathbb{N}$. Then $1-U_j=S_j-C_j$ and since $C_j\leq S_j$, it holds $U_j\in \{0,1\}$. Now let $i\in \mathbb{N}$ and $U_{1:i}=(U_1,\ldots,U_{i})^T\in \{0,1\}^{i}$ be arbitrary but fixed. 
With this, \eqref{adavora} is equivalent to
\begin{align}
 F_i(U_{1:i})&\coloneqq\sum_{j=1}^{i} \left(\alpha \gamma_j + \sum_{k=1}^{j-1} g_{k,j} U_k \alpha_k(U_{1:(k-1)}) \frac{1}{\tau_k-\lambda_k}\right)(1-U_j) \leq \alpha. \label{eq: F(C)}\end{align}
Note that we only wrote the dependence of $\alpha_k$ on $U_{1:(k-1)}=(U_1,\ldots,U_{k-1})^T$, although the parameters $\lambda_k$ and $\tau_k$ could depend on it as well. That is, because these parameters could also be fixed, meaning if we change the $U_{1:(k-1)}$ they would still be valid parameters for an ADDIS-Graph. In contrast, the $\alpha_k$ changes by definition. It is difficult to show the validity of \eqref{eq: F(C)} directly. However, we will see that there exists $\tilde{U}_{1:i}\in \{0,1\}^{i}$ that obviously fulfil $F_i(\tilde{U}_{1:i})\leq \alpha$. Therefore, the idea is to determine such a $\tilde{U}_{1:i}$ that additionally satisfies $F_i(U_{1:i}) \leq F_i(\tilde{U}_{1:i})$.

Let $l=\max\{j\in \{1,\ldots,i\}:U_j=1\}$ (we set $\max(\emptyset)=0$) and $U_{1:i}^{l}=(U_1^l,\ldots,U_{i}^l)^T$, where $U_j^l=U_j$ for all $j\neq l$ and $U_l^l=0$. We assume that $l>0$ (if $l=0$, we later see $F_i(U_{1:i})\leq \alpha$ anyway). In the next step we want to show that $F_i(U_{1:i})\leq F_i(U_{1:i}^{l})$. For shorter notation we write $\alpha_j=\alpha_j(U_{1:(j-1)})$ and $\alpha_j^l=\alpha_j(U_{1:(j-1)}^{l})$. Since for all $j\leq i$: $U_j^l=U_j$  ($j\neq l$), $U_j^l=0$ ($j\geq l$), $U_j=0$ ($j\geq l+1$) and $ \alpha_j^l=\alpha_j$ ($j\leq l$), we have:

\begin{align*}& F_i(U_{1:i}^{l})-F_i(U_{1:i})\\ &= \sum_{ j=1}^{ i} \alpha \gamma_j (1-U_j^l) -\sum_{ j=1}^{ i} \alpha \gamma_j (1-U_j) + \sum_{ j=1}^{ i} \left( \sum_{k=1}^{j-1} g_{k,j} U_k^l \alpha_k^l \frac{1}{\tau_k-\lambda_k}\right)(1-U_j^l) \\
&- \sum_{ j=1}^{ i} \left( \sum_{k=1}^{j-1} g_{k,j} U_k \alpha_k \frac{1}{\tau_k-\lambda_k}\right)(1-U_j)\\
&= \alpha \gamma_l + \sum_{ j=1}^{ i} \left( \sum_{k=1}^{j-1} g_{k,j} U_k^l \alpha_k^l \frac{1}{\tau_k-\lambda_k}\right)(1-U_j^l) - \sum_{ j=1}^{ i} \left( \sum_{k=1}^{j-1} g_{k,j} U_k \alpha_k \frac{1}{\tau_k-\lambda_k}\right)(1-U_j) \\
&= \alpha \gamma_l + \sum_{k=1}^{l-1} g_{k,l} U_k^l \alpha_k^l \frac{1}{\tau_k-\lambda_k} + \sum_{j=l+1}^{i}  \sum_{k=1}^{l-1} g_{k,j} U_k^l \alpha_k^l \frac{1}{\tau_k-\lambda_k} \\
&- \sum_{j=l+1}^{i} \sum_{k=1}^{l} g_{k,j} U_k \alpha_k \frac{1}{\tau_k-\lambda_k} \\
&= \alpha \gamma_l + \sum_{k=1}^{l-1} g_{k,l} U_k \alpha_k \frac{1}{\tau_k-\lambda_k} - \sum_{j=l+1}^{i} g_{l,j} \alpha_l \frac{1}{\tau_l-\lambda_l} \\
&\geq \alpha \gamma_l + \sum_{k=1}^{l-1} g_{k,l} U_k \alpha_k \frac{1}{\tau_k-\lambda_k} -  \alpha_l \frac{1}{\tau_l-\lambda_l}\stackrel{Def.  1}{=}0,
\end{align*}
where we used in the inequality that the sequence $(g_{l,j})_{j=l+1}^{\infty}$ is non-negative and sums to at most $1$ for all $l \in \mathbb{N}$.

Since the $U_{1:i}\in \{0,1\}^{i}$ was arbitrary,
this shows $F_i(U_{1:i})\leq F_i(U_{1:i}^0)$ for all $U_{1:i}\in \{0,1\}^{i}$, where 
 $U_{1:i}^0=(0,\ldots, 0)^T\in \{0,1\}^{i}$. Next, we deduce that $F_i(U_{1:i}^0) \leq \alpha$ and conclude the proof. For this, just recognize that $U_{1:i}^0$ means $U_j=0$ for all $j\leq i$. Hence, we obtain
\begin{align*}
F_i(U_{1:i}^0)
=\sum_{j=1}^{i} \alpha \gamma_j \leq \alpha. 
\end{align*}
\end{proof}

\begin{sloppypar}\begin{proof}[Proof of Proposition 1]

Every online procedure satisfying equation (1) is a sequence of non-negative random variables $(\alpha_i)_{i\in \mathbb{N}}$, where $\alpha_i$ is measurable with respect to $\mathcal{G}_{i-1}$, such that 
\begin{align}\sum_{j\leq i} \frac{\alpha_j}{\tau_j-\lambda_j}(S_j-C_j)\leq \alpha \quad \text{for all } i\in \mathbb{N}. \label{eq:proof_general}\end{align} 
Note that $\alpha_i$ is fully determined through $P_1,\ldots, P_{i-1}$. Hence, pessimistic assumptions about $S_i=\mathbbm{1}\{P_i\leq \tau_i\}$ and $C_i=\mathbbm{1}\{P_i\leq \lambda_i\}$ need to be made at step $i \in \mathbb{N}$ in order to satisfy equation \eqref{eq:proof_general}. Consequently, condition (1) is equivalent to 
\begin{align}0\leq \alpha_i \leq  (\tau_i-\lambda_i) \left(\alpha- \sum_{j\leq i-1} \frac{\alpha_j}{\tau_j-\lambda_j}(S_j-C_j)\right) \quad \text{for all } i\in \mathbb{N}. \label{eq:proof_general_2}\end{align}

Let $i\in \mathbb{N}$ be arbitrary but fixed. In addition, let $(\alpha_j)_{j<i}$ be levels obtained by an ADDIS-Graph with parameters $(\gamma_j)_{j<i}$ and $(g_{k,j})_{k<j<i}$. We want to prove that 
\begin{align*}
    \alpha_i(\gamma_i,(g_{j,i})_{j<i}))=(\tau_i-\lambda_i)\left(\alpha \gamma_i + \sum_{j=1}^{i-1} g_{j,i}(C_j-S_j+1)  \frac{\alpha_j}{\tau_i-\lambda_j}\right),
\end{align*}
where $\gamma_i\in \left[0,1-\sum_{j=1}^{i-1} \gamma_j\right]$ and $g_{j,i}\in \left[0,1-\sum_{k=j+1}^{i-1} g_{j,k}\right]$, $j\in \{1,\ldots,i-1\}$, can take any value in the interval $\left[0,(\tau_i-\lambda_i) \left(\alpha- \sum_{j\leq i-1} \frac{\alpha_j}{\tau_j-\lambda_j}(S_j-C_j)\right)\right]$. Since $\alpha_i$ is continuous in $\gamma_i$ and $(g_{j,i})_{j<i}$, it is sufficient to show that
$\alpha_i(0,(0)_{j<i})=0$ and $\alpha_i\left(1-\sum_{j=1}^{i-1} \gamma_j,\left(1-\sum_{k=j+1}^{i-1} g_{j,k}\right)_{j<i}\right)=(\tau_i-\lambda_i) \left(\alpha- \sum_{j\leq i-1} \frac{\alpha_j}{\tau_j-\lambda_j}(S_j-C_j)\right)$. The first equation follows immediately, hence we only need to show the second (we set $U_j=C_j-S_j+1$ for all $j\in \mathbb{N}$):
\begin{align*}
    & \text{ } \alpha_i\left(1-\sum_{j=1}^{i-1} \gamma_j,\left(1-\sum_{k=j+1}^{i-1} g_{j,k}\right)_{j<i}\right)-(\tau_i-\lambda_i) \left(\alpha- \sum_{j=1}^{i-1} \frac{\alpha_j}{\tau_j-\lambda_j}(1-U_j)\right) \\ &=
     (\tau_i-\lambda_i)\left(\alpha \left(1-\sum_{j=1}^{i-1} \gamma_j\right) + \sum_{j=1}^{i-1} \left(1-\sum_{k=j+1}^{i-1} g_{j,k} \right) U_j  \frac{\alpha_j}{\tau_i-\lambda_j}\right)  \\&- (\tau_i-\lambda_i) \left(\alpha- \sum_{j=1}^{ i-1} \left(\alpha \gamma_j + \sum_{k=1}^{j-1} g_{k,j} U_k  \frac{\alpha_k}{\tau_k-\lambda_k}  \right)(1-U_j)\right) \\
     &= (\tau_i-\lambda_i)\left(-\alpha \sum_{j=1}^{i-1} \gamma_j U_j + \sum_{j=1}^{i-1}U_j  \frac{\alpha_j}{\tau_i-\lambda_j} -  \sum_{j=1}^{i-1}\sum_{k=j+1}^{i-1} g_{j,k} U_j  \frac{\alpha_j}{\tau_i-\lambda_j}\right. \\
     & \left. + \sum_{j=1}^{ i-1} \sum_{k=1}^{j-1} g_{k,j} U_k  \frac{\alpha_k}{\tau_k-\lambda_k}
     - \sum_{j=1}^{ i-1} \left(\sum_{k=1}^{j-1} g_{k,j} U_k  \frac{\alpha_k}{\tau_k-\lambda_k}\right) U_j\right) \\
     &= (\tau_i-\lambda_i)\left(\sum_{j=1}^{i-1}U_j  \frac{\alpha_j}{\tau_i-\lambda_j} -  \sum_{j=1}^{i-1}U_j \left(\alpha  \gamma_j+\sum_{k=1}^{j-1} g_{k,j} U_k  \frac{\alpha_k}{\tau_k-\lambda_k}\right)\right)=0.
\end{align*}
Therefore, if some online Procedure $(\tilde{\alpha}_i)_{i\in \mathbb{N}}$ satisfying condition (1) is given, we can choose the parameters $(\gamma_i)_{i\in \mathbb{N}}$ and $(g_{j,i})_{j\in \mathbb{N}, i>j}$ such that for the individual significance levels of the ADDIS-Graph $(\alpha_i)_{i\in \mathbb{N}}$ holds $\alpha_i=\tilde{\alpha}_i$ for all $i\in \mathbb{N}$.  
\end{proof}\end{sloppypar}

\begin{proof}[Proof of Lemma 1]

% Let weights $(g_{j,i})_{j\in \mathbb{N}, i>j}$ be given. We make the following definitions
% \begin{align*}
% \alpha_i^* &= (\tau_i-\lambda_i)\left(\alpha \gamma_i + \sum_{j=1}^{i-1} g_{j,i} (C_j-S_j+1)  \frac{\alpha_j^*}{\tau_j-\lambda_j}\right)\\
% \alpha_i^{\text{spend}} &= (\tau_i-\lambda_i)\left(\alpha \gamma_i + \sum_{j=1}^{i-L_i-1} g_{j,i} (C_j-S_j+1)  \frac{\alpha_j^*}{\tau_j-\lambda_j}\right)\\
% \alpha_i &= (\tau_i-\lambda_i)\left(\alpha \gamma_i + \sum_{j=1}^{i-L_i-1} g_{j,i}^* (C_j-S_j+1)  \frac{\alpha_j}{\tau_j-\lambda_j}\right).
% \end{align*}
% $\alpha_i^*$ can be interpreted as the level that would be obtained by an ADDIS-Graph under independence, $\alpha_i^{\text{spend}}$ will be used to show how the $\text{ADDIS-Spending}_{\text{local}}$ can be obtained as a specific $\text{ADDIS-Graph}_{\text{local}}$ and $\alpha_i$ is the usual definition of the $\text{ADDIS-Graph}_{\text{local}}$ for some local dependence adjusted weights $(g_{j,i}^*)_{j\in \mathbb{N}, i>j}$ that depend on $(g_{j,i})_{j\in \mathbb{N}, i>j}$. 

In the following we write $\tilde{\alpha}_i^{\text{ind}}=\frac{\alpha_i^{\text{ind}}}{\tau_i-\lambda_i}$, $\tilde{\alpha}_i^{\text{loc}}=\frac{\alpha_i^{\text{loc}}}{\tau_i-\lambda_i}$ and $U_j=C_j-S_j+1$ to reduce notation. Let $g_{j,i}=\frac{\gamma_{t(j)+i-j-1}-\gamma_{t(j)+i-j}}{\gamma_{t(j)}}$, $i>j$, where $t(j)=1+\sum_{k<j} (1-U_k)$. Obviously, $\tilde{\alpha}_1^{\text{ind}}=\alpha \gamma_{t(1)}$. Now assume $\tilde{\alpha}_j^{\text{ind}}=\alpha \gamma_{t(j)}$ for all $j<i$. Thus, we have 
\begin{align*} \tilde{\alpha}_i^{\text{ind}}&=\alpha \gamma_i + \alpha \sum_{j=1}^{i-1} U_j (\gamma_{i-j+\sum_{k<j}(1-U_k)}-\gamma_{i-j+1+\sum_{k<j}(1-U_k)}) \\
&=\alpha \gamma_i + \alpha \sum_{j=t(i)}^{i-1}  (\gamma_{j}-\gamma_{j+1}) =\alpha \gamma_{t(i)}.   \end{align*}
With this, we can show that the $\text{ADDIS-Spending}_{\text{local}}$ can obtained by $\alpha_i^{\text{loc}}$ with the same choice of weights
\begin{align*} \tilde{\alpha}_i^{\text{loc}}&=\alpha \gamma_i + \alpha \sum_{j=1}^{i-L_i-1} U_j (\gamma_{i-j+\sum_{k<j}(1-U_k)}-\gamma_{i-j+1+\sum_{k<j}(1-U_k)})  \\
&=\alpha \gamma_i + \alpha \sum_{j=t(i)^{\text{loc}}}^{i-1}  (\gamma_{j}-\gamma_{j+1}) =\alpha \gamma_{t(i)^{\text{loc}}},   \end{align*}
where $t(i)^{\text{loc}}=1+L_i+\sum_{j=1}^{i-L_i-1} (S_j-C_j)$.
\end{proof}

 \begin{proof}[Proof of Proposition 2]
 Let $g_{j,i}=\frac{\gamma_{t(j)+i-j-1}-\gamma_{t(j)+i-j}}{\gamma_{t(j)}}$, $i>j$, where $t(j)=1+\sum_{k<j} (1-U_k)$ and $(g_{j,i}^*)_{j\in \mathbb{N}, i>j}$ be defined as in Algorithm \ref{alg:1}. 

\begin{algorithm}
\caption{Local dependence adjusted weights for uniform improvement } \label{alg:1}
\begin{algorithmic}
\State $g_{j,i}^* \gets g_{j,i} \text{ } \forall j\in \mathbb{N}, i>j$
\For{$j=1,2,\ldots$}
\For{$i=j+1,j+2,\ldots$}
\If{$i-L_i \leq j$} 
    \State $g_{j,i}^- \gets g_{j,i}^*$
    \State $g_{j,i}^* \gets 0$
    \For{$k>i$}
    $g_{j,k}^*\gets g_{j,k}^* + g_{j,i}^- g_{i,k}$
    \EndFor  
\Else
    \State $g_{j,i}^- \gets \sum_{l=i-L_i}^{i-1} g_{l,i} g_{j,l}^-$
    \State $g_{j,i}^* \gets g_{j,i}^*- g_{j,l}^-$
    \For{$k>i$}
    $g_{j,k}^*\gets g_{j,k}^* + g_{j,i}^- g_{i,k}$
    \EndFor
\EndIf 
\EndFor 
\EndFor
\end{algorithmic}
\end{algorithm}

The weight $g_{j,i}^-$ defined in Algorithm \ref{alg:1} can be interpreted as the part of $g_{j,i}^*$ that cannot be used at step $i$ due to local dependence and thus is distributed to the future weights $g_{j,k}^*$, $k>i$, according to the weights $g_{i,k}$. In case of $i-L_i\leq j$, $H_i$ is not allowed to use any significance level of $H_j$, which is why we set $g_{j,i}^-=g_{j,i}^*$ and thus $g_{j,i}^*=0$. This ensures that $g_{j,i}^*=0$ for all $j\in \mathcal{X}_i$, as required in Definition 2. Setting $g_{j,i}^-=\sum_{l=i-L_i}^{i-1} g_{l,i} g_{j,l}^-$ in case of $i-L_i> j$ additionally ensures that $g_{j,i}^*$ solely depends on $(g_{l,k})_{l\leq i-L_i, k>l}$, which is measurable with respect to $\sigma(P_1,\ldots,P_{i-L_i-1})=\mathcal{G}_{-\mathcal{X}_i}$. Furthermore, note that the sum of the $(g_{j,i}^*)_{i\geq j+1}$ is less or equal than the sum of $(g_{j,i})_{i\geq j+1}$ for each $j\in \mathbb{N}$ and since $$\sum_{i=j+1}^{\infty} g_{j,i}= \frac{1}{\gamma_{t(j)}}\sum_{i=j+1}^{\infty} \gamma_{t(j)+i-j-1}-\gamma_{t(j)+i-j}=1,$$
$(g_{j,i}^*)_{j\in \mathbb{N},i>j}$ can be used in the $\text{ADDIS-Graph}_{\text{conf}}$ (Definition 2).
 In the following, we show that the $\text{ADDIS-Graph}_{\text{conf}}$ with this choice of $(g_{j,i}^*)_{j\in \mathbb{N},i>j}$ leads to a uniform improvement over $\text{ADDIS-Spending}_{\text{local}}$.

 We need to show that $\tilde{\alpha}_i^{\text{loc}}$, defined in the proof of Lemma 1, is less or equal than $\tilde{\alpha}_i=\alpha \gamma_i + \sum_{j=1}^{i-L_i-1} g_{j,i}^{*} (C_j-S_j+1) \alpha_j$  for all $i\in \mathbb{N}$. For this, we define $g_{j,i}^{+,\text{loc}}$ and $g_{j,i}^{+}$ as the proportion of $\alpha \gamma_j$ that is shifted to $\tilde{\alpha}_i^{\text{loc}}$ and $\tilde{\alpha}_i$, respectively, in case of $P_j \leq \lambda_j$ or $P_j>\tau_j$. Hence, it is sufficient to show that $g_{j,i}^{+,\text{loc}}\leq g_{j,i}^{+}$ for all $j\in \mathbb{N}$ and $i>j$. Let $U_j=(C_j-S_j+1)$, then  $g_{j,i}^{+,\text{loc}}$ and $g_{j,i}^{+}$ can be calculated by the Algorithms \ref{alg:2} and \ref{alg:3}, respectively. Since $U_j\leq 1$, we have $g_{j,i}^{+}\geq g_{j,i}^{\text{spend},+}$ and the assertion follows.

\begin{algorithm}
\caption{Calculation of $g_{j,i}^{+,\text{loc}}$ for  $j\in \mathbb{N}$} \label{alg:2}
\begin{algorithmic}
\State $g_{j,i}^{+,\text{loc}} \gets g_{j,i} \text{ } \forall i>j$
\For{$i=j+1,j+2,\ldots$}
\If{$i-L_i\leq j$} 
    \State $g_{j,i}^- \gets g_{j,i}^{+,\text{loc}}$
    \State $g_{j,i}^{+,\text{loc}} \gets 0$
    \For{$k>i$}
    $g_{j,k}^{+,\text{loc}}\gets g_{j,k}^{+,\text{loc}} + g_{j,i}^- g_{i,k} U_i$
    \EndFor  
\Else
    \State $g_{j,i}^- \gets \sum_{l=i-L_i}^{i-1} g_{l,i} g_{j,l}^- U_l + \sum_{l=i-L_i}^{i-1} g_{l,i} g_{j,l}^+ U_l$
    \State $g_{j,i}^{+,\text{loc}} \gets g_{j,i}^{+,\text{loc}}- g_{j,l}^-$
    \For{$k>i$}
    $g_{j,k}^{+,\text{loc}}\gets g_{j,k}^{+,\text{loc}} + g_{j,i}^- g_{i,k} U_i + g_{j,i}^+ g_{i,k} U_i$
    \EndFor
\EndIf 
\EndFor 
\end{algorithmic}
\end{algorithm}

\begin{algorithm}
\caption{Calculation of $g_{j,i}^{+}$ for $j\in \mathbb{N}$} \label{alg:3}
\begin{algorithmic}
\State $g_{j,i}^{+} \gets g_{j,i} \text{ } \forall i>j$
\For{$i=j+1,j+2,\ldots$}
\If{$i < d_j$} 
    \State $g_{j,i}^- \gets g_{j,i}^{+}$
    \State $g_{j,i}^{+} \gets 0$
    \For{$k>i$}
    $g_{j,k}^{+}\gets g_{j,k}^{+} + g_{j,i}^- g_{i,k}$
    \EndFor  
\Else
    \State $g_{j,i}^- \gets \sum_{l=i-L_i}^{i-1} g_{l,i} g_{j,l}^- + \sum_{l=i-L_i}^{i-1} g_{l,i} g_{j,l}^+ U_l$
    \State $g_{j,i}^{+} \gets g_{j,i}^{+}- g_{j,l}^-$
    \For{$k>i$}
    $g_{j,k}^{+}\gets g_{j,k}^{+} + g_{j,i}^- g_{i,k} + g_{j,i}^+ g_{i,k} U_i$
    \EndFor
\EndIf 
\EndFor 
\end{algorithmic}
\end{algorithm}

 \end{proof}

\begin{proof}[Proof of Theorem \ref{theo:corr_princ}]
     % We first show weak control, strong control is then obtained by the subset pivotality condition. Assuming $I_0=\mathbb{N}$, we have
     % \begin{align*}
     % \text{FWER}(i)&=\left(\bigcup\limits_{j=1}^{ i} \{P_i\leq \alpha_i\}\right) \\
     % &= \sum_{j=1}^{i} \mathbb{P}\left( \bigcap\limits_{k\in \{1,\ldots,j-1\}} \{P_k>\alpha_k\} \cap \{P_j\leq \alpha_j\}    \right) \\
     % &\leq  \sum_{j=1}^{i} \mathbb{P}\left( \bigcap\limits_{k\in \{j-L_j,\ldots,j-1\}, P_k>\lambda} \{P_k>\alpha_k\} \cap \{P_j\leq \alpha_j\}    \right) \\
     % &=\sum_{j=1}^{i} \mathbb{E} \left[\mathbb{P}\left( \bigcap\limits_{k\in \{j-L_j,\ldots,j-1\}, P_k>\lambda} \{P_k>\alpha_k\} \cap \{P_j\leq \alpha_j\}\Bigg \vert \mathcal{G}_{j-L_j-L_{j-L_j}-1}    \right)   \right] \\
     % &\leq \sum_{j=1}^{i} \mathbb{E}\left[\alpha_j^{\lambda}\right] \\
     % &\leq \sum_{j=1}^{i} \mathbb{E}\left[\alpha_j^{\lambda} \mathbb{E} \left(  \frac{1-C_j}{1-\lambda} \bigg \vert \mathcal{G}_{j-L_j-L_{j-L_j}-1} \right) \right] \\
     % &= \mathbb{E}\left[ \sum_{j=1}^{i} \alpha_j^{\lambda} \frac{1-C_j}{1-\lambda} \right] \leq \alpha.
     % \end{align*}

First, note that
\begin{align*}
     \text{FWER}(i)&=\left(\bigcup\limits_{j\leq i, j\in I_0} \{P_i\leq \alpha_i\}\right) \\
     &\leq \sum_{j\leq i, j\in I_0} \mathbb{P}\left( \bigcap\limits_{k\in B_{b_j}, k<j, k\in I_0} \{P_k>\alpha_k\} \cap \{P_j\leq \alpha_j\}    \right) \\
     &\leq  \sum_{j\leq i, j\in I_0} \mathbb{P}\left( \bigcap\limits_{k\in B_{b_j}, k<j, k\in I_0, C_k=0} \{P_k>\alpha_k\} \cap \{P_j\leq \alpha_j\}    \right) \\
     &=\sum_{j\leq i, j\in I_0} \mathbb{E} \left[\mathbb{P}\left( \bigcap\limits_{k\in B_{b_j}, k<j, k\in I_0, C_k=0} \{P_k>\alpha_k\} \cap \{P_j\leq \alpha_j\}\Bigg \vert \mathcal{F}_{b_j-1}    \right)   \right] \\
     &= \sum_{j\leq i, j\in I_0} \mathbb{E}\left[\alpha_j^{c,I_0}\right] \\
     &\leq \sum_{j\leq i, j\in I_0} \mathbb{E}\left[\alpha_j^{c,I_0} \mathbb{E} \left(  \frac{1-C_j}{1-\lambda_{b_j}} \bigg \vert \mathcal{F}_{b_j-1} \right) \right] \\
     &= \mathbb{E}\left[ \sum_{j\leq i, j\in I_0} \alpha_j^{c,I_0} \frac{1-C_j}{1-\lambda_{b_j}} \right],  
     \end{align*}
     where $\alpha_j^{c,I_0}=\mathbb{P}\left( \bigcap\limits_{k\in B_{b_j}, k<j, k\in I_0, C_k=0} \{P_k>\alpha_k\} \cap \{P_j\leq \alpha_j\}\Bigg \vert \mathcal{F}_{b_j-1}    \right)$.
Using the subset pivotality condition, we obtain
\begin{align*}
    \sum_{j\leq i, j\in I_0} \alpha_j^{c, I_0} \frac{1-C_j}{1-\lambda_{b_j}} &=  \sum_{j=1}^{b_i} \frac{1}{1-\lambda_{b_j}} \mathbb{P}\left( \bigcup\limits_{k\in B_{b_j}, k\leq i, k\in I_0, C_k=0} \{P_k\leq \alpha_k\} \Bigg  \vert \mathcal{F}_{b_j-1}    \right) \\
    &\leq  \sum_{j=1}^{b_i} \frac{1}{1-\lambda_{b_j}} \mathbb{P}_{H_\mathbb{N}}\left( \bigcup\limits_{k\in B_{b_j}, k\leq i,  C_k=0} \{P_k\leq \alpha_k\} \Bigg  \vert \mathcal{F}_{b_j-1}    \right) \\
    &=  \sum_{j=1}^{i}  \frac{1}{1-\lambda_{b_j}} \alpha_j^{c} (1-C_j) \leq \alpha.
\end{align*}

% where $\alpha_j^{I,\lambda}\coloneqq \mathbb{P}\left( \bigcap\limits_{k\in B_{b_j}, k<j, j\in I, C_k=0} \{P_k>\alpha_k\} \cap \{P_j\leq \alpha_j\}\Bigg \vert H_{I}, \mathcal{F}_{b_j-1}    \right)$ for some $I\subseteq \mathbb{N}$ and $j \in I$. We now show that $\sum_{j\leq i, j\in I} \alpha_j^{\lambda,I} \frac{1-C_j}{1-\lambda} \leq \sum_{j\leq i, j\in I\cup \{l\}} \alpha_j^{\lambda,I\cup \{l\}} \frac{1-C_j}{1-\lambda}$ for some $l\notin I$ using the subset pivotality condition to conclude the proof. If $C_l=1$, the assertion follows immediately. Thus, we assume $C_l=0$.
% \begin{align*}
%  & \sum_{j\leq i, j\in I} \alpha_j^{\lambda,I} \frac{1-C_j}{1-\lambda} \\ &=  \frac{1}{1-\lambda} \sum_{j\leq i, j\in I} \mathbb{P}\left( \bigcap\limits_{k\in \{j-L_j,\ldots,j-1\}, j\in I, C_k=0} \{P_k>\alpha_k\} \cap \{P_j\leq \alpha_j\}\Bigg \vert H_{I}, \mathcal{G}_{j-L_j-L_{j-L_j}-1}    \right) (1-C_j)  
% \end{align*}
% NOT FINISHED

 \end{proof}

\begin{proof}[Proof of Theorem \ref{theo:graph_corr}]
    Obviously, $\alpha_i$ is measurable regarding $\mathcal{F}_{b_i-1}$. Now let $i\in \mathbb{N}$ be fix. We define $$A(k)\coloneqq \sum_{j=1}^{i-k} \frac{\alpha_j^{c} + (\alpha_j-\alpha_j^{c}) \sum_{l=i-k+1}^{i} g_{j,l}^*}{1-\lambda_{b_j}} (1-C_j) + \frac{\alpha_j \sum_{l=i-k+1}^{i} g_{j,l}^*}{1-\lambda_{b_j}} C_j. $$
    For $k\in \{0,\ldots,i-1\}$, we can write 
    \begin{align*}A(k)&=\sum_{j=1}^{i-k} \frac{\alpha_j^{c} + (\alpha_j-\alpha_j^{c}) \sum_{l=i-k+1}^{i} g_{j,l}^*}{1-\lambda_{b_j}} (1-C_j) + \frac{\alpha_j \sum_{l=i-k+1}^{i} g_{j,l}^*}{1-\lambda_{b_j}} C_j \\
    &\leq \frac{\alpha_{i-k}}{1-\lambda_{b_{i-k}}} + 
    \sum_{j=1}^{i-k-1} \frac{\alpha_j^{c} + (\alpha_j-\alpha_j^{c}) \sum_{l=i-k+1}^{i} g_{j,l}^*}{1-\lambda_{b_j}} (1-C_j) + \frac{\alpha_j \sum_{l=i-k+1}^{i} g_{j,l}^*}{1-\lambda_{b_j}} C_j \\
    &= \alpha \gamma_{i-k} + A(k+1).
    \end{align*}
 Since $A(i)=0$, we especially have $$A(0)=\sum_{j=1}^{i} \frac{\alpha_j^{c}}{1-\lambda_{b_j}}(1-C_j) \leq \alpha \sum_{j=1}^{i} \gamma_j \leq \alpha.$$  
\end{proof}

\begin{proof}[Proof of Theorem \ref{theo:addis_graph_fdr}]
    Let $\alpha_j^0 = W_0 \gamma_j  + \sum_{k=1}^{j-1} h_{k,j}^* R_k [\alpha K_k + (\alpha-W_0)K_k^c]$ for all $j\in \mathbb{N}$. Note that 
    \begin{align*}
        \sum_{j=1}^i \alpha_j^0 &=  \sum_{j=1}^iW_0 \gamma_j  +  \sum_{j=1}^i \sum_{k=1}^{j-1} h_{k,j}^* R_k [\alpha K_k + (\alpha-W_0)K_k^c]  \\
        & \leq W_0 +\sum_{k=1}^{i-1} R_k [\alpha K_k + (\alpha-W_0)K_k^c] \sum_{j=k+1}^i h_{k,j}^* \\
        & \leq W_0 +\sum_{k=1}^{i-1} R_k [\alpha K_k + (\alpha-W_0)K_k^c] \\
        &= W_0 + (\alpha-W_0) K_i + \alpha(|R(i-1)|-1) K_i \\ &\leq \alpha(|R(i)|\lor 1)
    \end{align*}
   With this, it can be shown that the $\text{ADDIS-Graph}_{\text{conf}}$ satisfies (1) in the same way as in Theorem 1 by replacing $\alpha \gamma_j$ with $\alpha_j^0$ on the left side in equation \eqref{eq: F(C)} and $\alpha$ with $\alpha(|R(i)|\lor 1)$ on the right side. Hence, if the null p-values are independent from each other and the non-nulls and $\alpha_i$, $\lambda_i$ and $\tau_i$ are monotonic functions of the past, then the FDR control follows immediately by Theorem 1 of Tian and Ramdas (2019) \citep{TR2}. Furthermore, the mFDR control for $\tau_i=1$, $i\in \mathbb{N}$, follows by Theorem 2 of Zrnic et al. (2020) \citep{ZRJ}. However, since we consider general $\tau_i\in (0,1]$, we provide a self-contained proof for mFDR control when \eqref{eq:FDR_ADDIS_principle} is fulfilled, which is very similar to the proofs by Zrnic et al. (2020) \citep{ZRJ} and Tian and Ramdas (2019) \citep{TR2}.

Consider 
\begin{align*}
\mathbb{E}[|V(i)|]&=\sum_{j\in I_0, j\leq i} \mathbb{E}[R_j] \\
&= \sum_{j\in I_0, j\leq i} \mathbb{E}[\mathbb{P}(P_j\leq \alpha_j|P_j\leq \tau_j,\mathcal{F}_{-\mathcal{X}_j}) \mathbb{P}(P_j\leq \tau_j|\mathcal{F}_{-\mathcal{X}_j})] \\
&\leq \sum_{j\in I_0, j\leq i} \mathbb{E}\left[\frac{\alpha_j}{\tau_j} \mathbb{P}(P_j\leq \tau_j|\mathcal{F}_{-\mathcal{X}_j})\right] \\
&\leq \sum_{j\in I_0, j\leq i} \mathbb{E}\left[\frac{\alpha_j}{\tau_j} \mathbb{P}(P_j\leq \tau_j|\mathcal{F}_{-\mathcal{X}_j}) \frac{\mathbb{P}(P_j> \lambda_j|P_j\leq \tau_j,\mathcal{F}_{-\mathcal{X}_j})}{1-\lambda_j/\tau_j}\right] \\
&= \sum_{j\in I_0, j\leq i} \mathbb{E}\left[\frac{\alpha_j}{\tau_j-\lambda_j} \mathbbm{1}\{\lambda_j< P_j \leq \tau_j \}\right] \\
&\leq \mathbb{E}[|R(i)|\lor 1] \alpha,
\end{align*} 
where the first and second inequality follow from the uniform validity of the null p-values and the third inequality from condition \eqref{eq:FDR_ADDIS_principle}.

\end{proof}
    
\end{appendix}

\end{document}